%% file: rhle.tex
\documentclass[runningheads]{llncs}
\usepackage[T1]{fontenc}
\usepackage{graphicx}

\usepackage{booktabs}
\usepackage{listings}
\usepackage{centernot} 
\usepackage[ruled,vlined,linesnumbered]{algorithm2e}
\usepackage{amsmath}
\usepackage{amssymb}
\usepackage{stmaryrd}
\usepackage{bbding}
\usepackage{booktabs}
\usepackage{enumitem}
\usepackage{float}
\usepackage{xcolor}

\definecolor{DPurple}{RGB}{97,38,103}

\usepackage{hyperref}
\hypersetup{
    colorlinks=true,
    linkbordercolor={white},
    linkcolor={DPurple},
    citecolor={DPurple}
}

\usepackage{fancyvrb}
\usepackage{makecell}
\usepackage{mathpartir}
\usepackage{multirow}
\usepackage{outlines}
\usepackage{xparse}
\usepackage{microtype}
\usepackage{IEEEtrantools} 
\usepackage{hhline}

\usepackage{wrapfig}
\usepackage{framed}
\usepackage{pifont}

\setlist[description]{leftmargin=\parindent,labelindent=\parindent}

\usepackage{courier}

\include{commands}

\def\sectionautorefname{Section}

\begin{document}

\title{RHLE: Modular Deductive Verification of Relational \AEH{}
       Properties}

\author{Robert Dickerson \and
        Qianchuan Ye \and
        Michael K. Zhang \and
        Benjamin Delaware}

\institute{Purdue University, West Lafayette, IN 47907, USA \\
          \email{rob@robd.io, ye202@purdue.edu, michael.k.zhang@alumni.purdue.edu, bendy@purdue.edu}}

\date{}
\maketitle
\vspace{-2em}
\begin{abstract}
  Hoare-style program logics are a popular and effective technique for
  software verification. Relational program logics are an instance of
  this approach that enables reasoning about relationships between the
  execution of two or more programs.  Existing relational program
  logics have focused on verifying that \emph{all} runs of a
  collection of programs do not violate a specified relational
  behavior. Several important relational properties, including
  refinement and noninterference, do not fit into this category, as
  they also mandate the \emph{existence} of specific desirable
  executions. This paper presents \rhle{}, a logic for verifying these
  sorts of relational \AEH{} properties. Key to our approach is a
  novel form of function specification that employs a variant of ghost
  variables to ensure that valid implementations exhibit certain
  behaviors. We have used a program verifier based on \rhle{} to
  verify a diverse set of relational \AEH{} properties drawn from the
  literature.
\end{abstract}

\input{sections/01-introduction}
\input{sections/02-funimp}
\input{sections/03-approximating-behaviors}
\input{sections/04-rhle}
\input{sections/05-verification}
\input{sections/06-evaluation}
\input{sections/07-related-work}
\input{sections/08-conclusion}

\section*{Acknowledgements}

We would like to thank Roopsha Samanta for her valuable input on
initial drafts of this work. We would also like to thank the anonymous
reviewers of this and previous iterations of this paper for their
much-appreciated feedback. This research was partially supported by
the National Science Foundation under Grant CCF-1755880 and by a grant
from the Purdue Research Foundation.

\def\bibfont{\normalsize}
\bibliography{bibliography}

\input{sections/99-appendix}

\end{document}

%% file: commands.tex
\definecolor{DOrange}{RGB}{145,20,20}
\definecolor{DPurple}{RGB}{79,21,95}

\definecolor{CBPurple}{HTML}{702A5A}
\definecolor{CBRed}{HTML}{EE442F}
\definecolor{CBBlue}{HTML}{63ABCE}

\newcommand{\parheading}[1]{\textit{#1}\,}

\newcommand{\AEH}{$\forall\exists$}

\newcommand{\FunImp}{\textsc{FunIMP}}

\newcommand{\esim}{\sim_{_{\exists}}}
\newcommand{\hltrip}[3]{S_\forall \vdash \left\{#1\right\}\ #2\ \{#3\}}
\newcommand{\semhltrip}[3]{S_\forall \models \left\{#1\right\}\ #2\ \{#3\}}

\newcommand{\hletrip}[3]{S_\exists \vdash \left[#1\right]\ #2\ \left[#3\right]_{\exists}}
\newcommand{\semhletrip}[3]{S_\exists \models \left[#1\right]\ #2\ \left[#3\right]_{\exists}}
\newcommand{\epre}[1]{\langle #1 \rangle \ }
\newcommand{\epost}[1]{\ \langle #1 \rangle}
\newcommand{\rhletrip}[4]{S_\forall, S_\exists \vdash \epre{#1} #2 \esim #3 \epost{#4}}
\newcommand{\rhletripnosigma}[4]{\epre{#1} #2 \esim #3 \epost{#4}}
\newcommand{\semrhletrip}[4]{S_\forall, S_\exists \models
\left\langle #1 \right\rangle #2 \esim #3 \left\langle#4 \right\rangle}

\newcommand{\seq}[2]{#1\textsf{; }#2}
\newcommand{\asgn}[2]{#1 \textsf{ := } #2}
\newcommand{\ite}[3]{\textsf{if } #1 \textsf{ then } #2 \textsf{ else } #3}
\newcommand{\while}[2]{\textsf{while } #1 \textsf{ do } #2 \textsf{ end}}
\newcommand{\whilenoend}[2]{\textsf{while } #1 \textsf{ do } #2}
\newcommand{\alt}{\ | \ }
\newcommand{\sskip}{\textsf{skip}}

\newcommand{\rhle}{\textsc{RHLE}}
\newcommand{\orhle}{\textsc{ORHLE}}

\newcommand{\productive}{\Downarrow_\exists}

\definecolor{darkblue}{HTML}{06038D}

\def\lstinlines|#1|{\text{\lstinline[basicstyle=\footnotesize]|#1|}}
\def\lstinlineGNI|#1|{\text{\lstinline[language=Java, style=Java,
mathescape=true]|#1|}}

\lstdefinestyle{rest}{
  basicstyle=\ttfamily,
  lineskip=-.05cm
}

\lstdefinestyle{Java}{
  language=Java,
  basicstyle=\ttfamily\small,
  keywordstyle=\bfseries\color{darkblue},
  lineskip=-.05cm,
  mathescape=true,
  morekeywords={assert, xor}
}

\lstdefinestyle{Python}{
  basicstyle=\ttfamily\small,
  lineskip=-.05cm
}

\lstdefinestyle{funimp}{
  basicstyle=\sffamily,  %
  breaklines=true, %
  columns=fullflexible, %
  keepspaces=true,
  escapechar=\&,
  literate=%
  {forall}{{$\forall$}}{1}
  {exists}{{$\exists$}}{1}
  {TOP}{{$\top$}}{1}
  {c_1}{{card1}}{1}
  {c_2}{{card2}}{1}
  {c_3}{{card3}}{1}
  {c_11}{{card1$_1$}}{1}
  {c_12}{{card2$_2$}}{1}
  {c_21}{{card2$_1$}}{1}
  {c_22}{{card2$_2$}}{1}
  {player_1}{{p$_1$}}{1}
  {player_2}{{p$_2$}}{1}
  {player_3}{{p$_3$}}{1}
  {hand_1}{{hand$_1$}}{1}
  {hand_2}{{hand$_2$}}{1}
  {<=}{{$\le$}}{1}
  {/\\}{{$\land$}}{1}
  {AND}{{$\land$}}{1}
  {==>}{{$\implies$}}{1}
  {->}{{$\to$}}{1}
  {++}{{$+\hspace{-3pt}+$}}{1}
  {<>}{{$\neq$}}{1}, %
  keywordstyle=\textbf, %
  morecomment=[s]{(*}{*)}, %
  moredelim=**[is][\color{StringIDs}]{@}{@},
  moredelim=[is][\SubScript]{SUB}{SUB},
  morekeywords={
     def,
     return,
   }
}

\lstdefinestyle{orhle}{
  basicstyle=\sffamily,  %
  breaklines=true, %
  columns=fullflexible, %
  keepspaces=true,
  escapechar=\&,
  basicstyle=\footnotesize\sffamily,
  keywordstyle=\bfseries,
  morecomment=[s]{(*}{*)}, %
  moredelim=**[is][\color{StringIDs}]{@}{@},
  moredelim=[is][\SubScript]{SUB}{SUB},
  morekeywords={
    end,
    ret,
    def,
    return,
    if,
    then,
    else,
    while,
    do,
    skip,
    prog,
    pre,
    post,
    endif,
    forall,
    exists,
    expected,
    aspecs,
    especs,
    call,
    fun
  }
}

\newcounter{daggerfootnote}

\newcommand{\apref}[1]{\renewcommand{\sectionautorefname}{Appendix}\autoref{#1}\renewcommand{\sectionautorefname}{Section}}

%% file: sections/01-introduction.tex
\section{Introduction}
\label{sec:Introduction}

Hoare-style program logics are a popular and effective verification
technique. Starting with Hoare's seminal paper~\cite{Hoare+69}, this
approach has been adapted to cover a variety of programming languages
and assertions~\cite{VST,Iris,Hoare-Java,Sep+Logic,Concur+Sep+Logic}.
These logics typically feature several pleasant properties: they can
be declaratively specified via a set of rules over the syntax of the
target programming language, they permit compositional reasoning over
individual program components, and they often admit effective
automated verification procedures. Most of these logics focus on
proving \textit{safety} properties of \textit{single} programs, i.e.,
that executing a program in a valid initial state never results in a
state violating a postcondition.

Not all program behaviors fall into this category, however. As one
example, consider the common scenario where a developer decides they
want to migrate a hand-rolled implementation of a function to one that
uses a third-party library. \autoref{fig:IntroRefinementExample}
gives a concrete example of this situation. The program on the left,
\lstinline|sample|$_1$, uses a random number generator to directly
sample a subset of an array. The program on the right,
\lstinline|sample|$_2$, opts to delegate the task to an external list
library which supports shuffling and constructing sublists.  While
\lstinline|sample|$_1$ works \textit{with replacement} (the same
elements may be sampled multiple times), \lstinline|sample|$_2$ works
\textit{without replacement} (an element may be sampled at most
once). In order to ensure that this change does not break things, the
developer may wish to verify that \lstinline|sample|$_2$ does not do
anything that \lstinline|sample|$_1$ could not, i.e., that the updated
function \emph{refines} the original. Notably, this refinement
property relates the behavior of \emph{multiple} programs. In
addition, it does not have the form of a standard safety property. The
developer does not want to enforce that \lstinline|sample|$_2$
produces \emph{every} permutation that the hand-rolled implementation
does; rather, they wish to ensure it does not start returning
previously impossible samples.

\begin{figure}[t!]
\begin{center}
\setlength{\tabcolsep}{8pt}
\begin{tabular}{c|c}
\begin{lstlisting}[style=Java]
int[] sample$_1$(int[] arr,
              int size) {
  assert(size <= arr.length);
  int[] samp = new int[size];
  for (i in [0..size]) {
    int j = randB(arr.length);
    samp[i] = arr[j];
  }
  return samp;
}
\end{lstlisting}
&
\begin{lstlisting}[style=Java]
int[] sample$_2$(int[] arr,
              int size) {
  assert(size <= arr.length);
  list = new List(arr);
  perm = list.permute();
  samp = perm.sublist(size);
  return samp.toArray();
}
\end{lstlisting}
\end{tabular}
\end{center}
\vspace{-1em}
\caption{An example migration of a function which randomly samples a
list of integers with replacement to a function which samples without
replacement. The original program (\lstinline|sample|$_1$) uses a
function which generates random numbers, while the migrated program
(\lstinline|sample|$_2$) uses a list abstraction with a
\lstinline{permute} operation.}
\label{fig:IntroRefinementExample}
\vspace{-1em}
\end{figure}

\begin{wrapfigure}{R}{\dimexpr 0.43\textwidth + 2\FrameSep + 2\FrameRule\relax}
\vspace{-3em}
\begin{framed}\centering
\begin{minipage}{\textwidth}
\vspace{-1em}
\begin{lstlisting}[language=Java, style=Java, mathescape=true]
int encode(int msg$^\texttt{H}$) {
  int key$^\texttt{H}$ = randB(MAX_INT);
  int enc$^\texttt{L}$ = msg$^\texttt{H}$ xor key$^\texttt{H}$;
  return (key$^\texttt{H}$, enc$^\texttt{L}$);
}
\end{lstlisting}
\vspace{-1em}
\end{minipage}
\end{framed}
\vspace{-3em}
\label{fig:IntroGNIExample}
\end{wrapfigure}

As another example, consider the \lstinline|encode| function on the
right which performs a simple xor cipher. This function takes a single
high-security argument, \lstinlineGNI|msg$^\texttt{H}$|, and returns a
pair of high-security and low-security results,
\lstinlineGNI|key$^\texttt{H}$| and
\lstinlineGNI|enc$^\texttt{L}$|, respectively. The
function encodes its argument by first generating a random key
(\lstinline|randB| returns a random value between 0 and its argument),
taking the \texttt{xor} of the key and the message, and finally
returning the key along with the encoded message.  The developer may
wish to guarantee an attacker can learn nothing about the secret
message given only the encoded message. Whether or not
\lstinlineGNI|encode| meets this \textit{generalized
  noninterference}~\cite{Mclean1996} property crucially depends on the
behavior of \lstinline|randB|: if the attacker knows this function
\emph{always} returns \lstinline|3|, for example, they can decipher
any encoded message. We can again frame this behavior as a relational
property between the executions of two programs (in this case calls to
\lstinlineGNI|encode| with arbitrary arguments
\lstinlineGNI|msg|$^\texttt{H}_1$ and
\lstinlineGNI|msg|$^\texttt{H}_2$): every execution of
\lstinlineGNI|encode(msg$^\texttt{H}_1$)| must have a corresponding
execution of \lstinlineGNI|encode(msg$^\texttt{H}_2$)| that returns
the same low-security encoded value.

In both examples, the desired behavior has the shape \emph{for all}
executions of some program, \emph{there exists} a corresponding
execution of a second program that is somehow related. Thus, we call
these properties \textit{relational \AEH{}} properties. While several
\emph{relational program logics} have been developed for reasoning
about the behavior of multiple
programs~\cite{Sousa+CHL,Benton+RHL,Barthe+pRHL}, all have focused on
relational \emph{safety} properties, i.e., that \emph{all} the final
states of multiple programs satisfy some relational postcondition.
Unfortunately, in the presence of nondeterminism, none of these logics
are capable of verifying relational \AEH{} properties such as
refinement and generalized noninterference. The need to reason about
nondeterminism naturally arises in the presence of external functions
like \lstinline|permute| in \autoref{fig:IntroRefinementExample},
where specifications are used to approximate the behavior of multiple
possible implementations.

This paper addresses this gap by introducing \rhle{}, a relational
program logic for reasoning about \AEH{} properties. Key to our
approach is a novel form of function specifications which approximate
the set of behaviors a valid implementation \emph{must} exhibit. These
specifications use a novel variant of ghost variables, which we call
\textit{choice variables}, that guarantee the existence of required
behaviors.  \rhle{} admits a modular reasoning principle, where any
properties verified against a set of function specifications continue
to hold whenever the program is linked to any satisfying
implementation. While techniques based on Constrained Horn
Clauses~\cite{Unno2021} and model checking~\cite{Lamport2021} have
recently been developed that are capable of reasoning about \AEH{}
properties, \rhle{} is, to the best of our knowledge, the first
Hoare-style program logic for doing so. We have used a verifier based
on \rhle{} to verify a range of \AEH{} properties including
refinement, noninterference (with and without delimited release),
semantic parameter usage, and flaky tests.

We begin by defining a core imperative language with function calls
(\autoref{sec:FunImp}) equipped with semantics for both over- and
under-approximating function behaviors
(\autoref{sec:ApproximatingBehaviors}). We next present \rhle{}, and a
corresponding verification algorithm for verifying \AEH{} properties
(\autoref{sec:Verification}). We evaluate our approach by applying an
implementation of this algorithm to verify a diverse set of relational
properties (\autoref{sec:Evaluation}). We conclude with an examination
of related work (\autoref{sec:RelatedWork}). We have formalized the
details of our approach in the Coq proof assistant; this development
is available in the supplementary materials of this paper. Our
verification tool and benchmark suite are also publicly
available~\cite{OrhleZenodo,OrhleBenchmarks}.

%% file: sections/02-funimp.tex
\section{The FunIMP Language}
\label{sec:FunImp}
\vspace{-.7em}
\begin{figure}[t!]
  \begin{equation*}
  \begin{aligned}[t]
    n & \in \mathbb{N}     \qquad   x, y \in \mathcal{V} \\
    f, g & \in \mathcal{N} \qquad   \phantom{x, }\sigma \in \mathcal{V} \rightarrow \mathbb{N} \\
    a & ::=~ n \alt x \alt a + a \alt a - a \alt a * a \\
    b & ::=~ \textsf{true} \alt \textsf{false} \\
      & \alt a = a \alt a < a \alt \neg b \alt b \wedge b \\
    \end{aligned}
    \qquad
    \begin{aligned}[t]
    s & ::= \sskip \alt \seq{s}{s} \\
          & \alt \ite{b}{s}{s} \\
          & \alt \while{b}{s}  \\
          & \alt x := a \alt x := \textsf{havoc} \alt x := f(\overline{a}) \\
    \mathit{FD} & ::=~ \textsf{\textbf{def} } f (\overline{x})\; \{s;
                  \textsf{\textbf{return} } a \}
  \end{aligned}
  \end{equation*}
  \vspace{-1em}
  \caption{Syntax of \FunImp{}.}
  \label{fig:Syntax}
    \vspace{-2em}
\end{figure}

We begin with the definition of \FunImp{}, a core imperative language
with function calls $x := f(\overline{a})$ and nondeterministic
variable assignment $x := \textsf{havoc}$. The full syntax of
\FunImp{} is presented in \autoref{fig:Syntax}. The calculus is
parameterized over disjoint sets of identifiers for program variables
$\mathcal{V}$ and function names $\mathcal{N}$. Functions have a fixed
arity. Function definitions consist of a sequence of statements
followed by an expression that computes the result of the
function. For brevity, we denote sequences $x_1, \ldots, x_n$ as
$\overline{x}$. For ease of presentation, we treat functions as
returning a single value, although it is straightforward to extend
\FunImp{} to allow for multiple return values:
$(x, y, \ldots) := f(\overline{a})$.  Our verification tool, ORHLE
(see \autoref{sec:Evaluation}), uses such an extension to model
functions which mutate their arguments.

The semantics of \FunImp{} programs are defined via a standard
big-step evaluation relation from initial to final program states.
States are mappings from variables to integers, and are usually
notated as $\sigma$. We write $[x \mapsto a]\sigma$ to refer to state
$\sigma$ updated with a mapping from $x$ to $a$. The evaluation rules
are parameterized over an \emph{implementation context}, a mapping $I
\in \mathcal{N} \rightarrow \mathit{FD}$ from function names to their
definitions, which is used to evaluate function calls:
\begin{mathpar}
  \inferrule*[right=ECall] {
    I(f) = \textsf{\textbf{def} } f(\overline{x})\, \{s; \textsf{\textbf{return} } e \} \\
    I \vdash \sigma, \overline{a} \Downarrow \overline{v} \\
    I \vdash [ \overline{x} \mapsto \overline{v}], s \Downarrow \sigma' \\
    I \vdash \sigma', e \Downarrow r \\
  }
  {
    I \vdash \sigma, y := f(\overline{a}) \Downarrow
    [y \mapsto r]\sigma
  }
\end{mathpar}
We use $\Downarrow$ for the evaluation relation of both expressions
and statements; $\sigma, e \Downarrow \sigma'$ holds when executing
$e$ on state $\sigma$ can result in state $\sigma'$. Since programs
may be nondeterministic, there may be multiple final states related to
a single initial state for a given program. Note that $\textsf{havoc}$
is the only source of nondeterminism when evaluating a \FunImp{}
program. The remaining evaluation rules for \FunImp{} are standard and
can be found in \apref{sec:FunImpSemantics}.

%% file: sections/03-approximating-behaviors.tex
\section{Approximating \FunImp{} Behaviors}
\label{sec:ApproximatingBehaviors}
\vspace{-.7em}

In order to modularly reason about relational \AEH{} properties, we
first present semantics for capturing the possible executions of a
\FunImp{} program in \emph{any} valid implementation context. In order
to account for both ``for all'' and ``there exists'' behaviors of
functions, we rely on two kinds of specifications. To reason about
\emph{all} possible executions of a valid implementation, i.e., a
standard \emph{safety} property, we use a \emph{universal}
specification. For guarantees about the \emph{existence} of certain
executions, we use an \emph{existential} specification.

\subsection{Universal Executions}
Both kinds of specifications are parameterized over an assertion
language $\mathcal{A}$ on program states and a mechanism for
judging when a state satisfies an assertion. We write
$\sigma \models P$ to denote that a state $\sigma$ satisfies the
assertion $P$. The universal specifications used to reason about
programs on the ``for all'' side of \AEH{} properties are written as
$\mathit{FA} ::= \textsf{ax}_{\forall} \ f (\overline{x})\; \{P\}
\{Q\}$, where $P \in \mathcal{A}$ is a precondition with free
variables in $\overline{x}$ and $Q \in \mathcal{A}$ is a postcondition
with free variables in $\overline{x} \cup \{ \rho \}$. The
postcondition uses the distinguished variable $\rho$ to refer to the
value returned by $f$. Universal specifications promise client
programs that the valid implementations of a function will only
evaluate to states satisfying the postcondition when evaluated in a
starting state that satisfies the precondition.

\begin{definition}[$\forall-Compatibility$]
  A function definition
  $\textsf{def } f (\overline{x}) \{s; \textsf{return } r \}$ is
  $\forall$\textit{-compatible} with a universal specification
  $\textsf{ax}_{\forall}\; f (\overline{x}) \{P\} \{Q\}$ if only
  values satisfying $Q$ may be returned whenever $f$ is called with
  arguments satisfying $P$:
\begin{align*}
  \forall \sigma, \sigma'.\; & (\sigma \models P)
                    \; \land \;
                    (I \vdash \sigma, s \Downarrow \sigma')
                    \; \land \;
                    (\sigma', r \Downarrow v)
                               \implies
                             ([\rho \mapsto v]\sigma \models Q)
\end{align*}
\end{definition}

\noindent We say that an implementation context $I$ is $\forall$-compatible with
a context of universal specifications
$S_\forall \in \mathcal{N} \rightarrow \mathit{FA}$ when every definition in
$I$ is $\forall$-compatible with the corresponding specification in
$S_\forall$.

To characterize the set of possible behaviors of a program under any
$\forall$-compatible implementation context, we define a new
\emph{overapproximate} semantics for \FunImp{}, $\Downarrow_\forall$.
The evaluation rules of this semantics are based on $\Downarrow$, but they
use a universal specification context, $S_\forall$, instead of an
implementation context, and replace \textsc{ECall} with the following
two evaluation rules:
\begin{mathpar}
  \inferrule*[right=ECall$_{\forall 1}$]
  {
    S_{\forall}(f) = \textsf{ax}_{\forall}\; f(\overline{x})\, \{P\}\, \{Q\} \\
    S_{\forall} \vdash \sigma, \overline{a} \Downarrow_\forall \overline{v} \\
    [ \overline{x} \mapsto \overline{v}] \models P  \\
    [ \rho \mapsto r, \overline{x} \mapsto \overline{v}] \models Q  \\
  }
  {
    S_{\forall} \vdash \sigma, y := f(\overline{a}) \Downarrow_\forall [y \mapsto r]\sigma
  }

  \inferrule*[right=ECall$_{\forall 2}$]
  {
    S_{\forall}(f) = \textsf{ax}_{\forall}\; f(\overline{x})\, \{P\}\, \{Q\} \\
    S_{\forall} \vdash \sigma, \overline{a} \Downarrow_\forall \overline{v} \\
    [ \overline{x} \mapsto \overline{v}] \not \models P  \\
  }
  {
    S_{\forall} \vdash \sigma, y := f(\overline{a}) \Downarrow_\forall [y \mapsto r]\sigma
  }
\end{mathpar}
\noindent The first rule states that if a function is called with
arguments satisfying its precondition, it will return a value
satisfying its postcondition; otherwise, the second rule states that
it can return \emph{any} value. The latter case allows the
overapproximate semantics to capture evaluations where a function is
called with arguments that do not meet its
precondition. \apref{sec:FunImpSemantics} includes a complete listing
of the $\Downarrow_\forall$ relation.

Any final state of a program evaluated
under an implementation context $I$ which is $\forall$-compatible with
$S_{\forall}$ can also be produced using $\Downarrow_\forall$ and
$S_\forall$. Appealing to this intuition, we call the evaluations of a
\FunImp{} program $p$ using $\Downarrow_{\forall}$ the
\emph{overapproximate executions} of $p$ under $S_\forall$.
\begin{theorem}
  \label{thm:ACompat}
  When run under an implementation context $I$ that is
  $\forall$-compatible with specification context $S_{\forall}$ and
  an initial state $\sigma$, a program $p$ will either diverge or
  evaluate to a state $\sigma'$ which is also the result of one of its
  overapproximate executions under $S_{\forall}$.
\end{theorem}

\subsection{Existential Executions}
Universal specifications approximate function calls on the ``for all''
side of \AEH{} properties by constraining what a compatible
implementation \emph{can} do. Existential specifications approximate
the ``there exists'' executions by describing the required values a
valid implementation \emph{must} be able to return. In order to
flexibly capture these behaviors, existential pre- and post-conditions
are indexed by a set of \textit{choice variables} $\overline{c}
\subseteq \mathcal{V}$. Each instantiation of these variables defines
a particular behavior that an implementation has to exhibit. The
syntax for writing an existential specification is:
$\mathit{FE} ::= \textsf{ax}_{\exists} \ f (\overline{x})\; [\overline{c}]\
\{P\} \{Q\}$.

We write $A[x/y]$ to denote the predicate $A$ with all free
occurrences of $x$ replaced with $y$. Intuitively, for any
instantiation $\overline{v}$ of choice variables $\overline{c}$, an
existential specification requires an implementation to produce at
least one value satisfying the specialized postcondition
$Q[\overline{v}/\overline{c}]$, when called with arguments that
satisfy the corresponding precondition
$P[\overline{v}/\overline{c}]$. This intuition is embodied in our
notion of compatibility for existential specifications:

\begin{definition}[$\exists$-Compatibility]
  A function definition
  $\textsf{def } f (\overline{x}) \{s; \textsf{return } r \}$ is
  $\exists$\textit{-compatible} with an existential specification
  $\textsf{ax}_{\exists}\; f (\overline{x}) [\overline{c}] \{P\}
  \{Q\}$ if, for every selection of choice variables $\overline{v}$,
  calling $f$ with arguments that satisfy
  $P[\overline{v}/\overline{c}]$ can return at least one value
  satisfying $Q[\overline{v}/\overline{c}]$:
\begin{align*}
  \forall \sigma, \overline{v}.\;
  & (\sigma \models P[\overline{v}/\overline{c}]) \implies
    \exists \sigma'.\; (I \vdash \sigma, s \Downarrow \sigma')
    \; \land \;
    (\sigma', r \Downarrow v)
    \land ([\rho \mapsto v]\sigma \models Q[\overline{v}/\overline{c}])
\end{align*}
\vspace{-3em} 
\end{definition}

\begin{figure}[t]
\begin{tabular}{c|c|c}
\begin{minipage}{.25\columnwidth}
\begin{lstlisting}[style=funimp, basicstyle=\small\sffamily, mathescape=true]
def randB(x) {
  skip;
  return 0
}
\end{lstlisting}
\end{minipage}
&
\begin{minipage}{.4\columnwidth}
\begin{lstlisting}[style=funimp, basicstyle=\small\sffamily, mathescape=true]
def randB(x) {
  r := havoc;
  while (x <= r) do r := r - x end;
  return r }
\end{lstlisting}
\end{minipage}&
  \begin{minipage}{.25\columnwidth}
\begin{lstlisting}[style=funimp, basicstyle=\small\sffamily, mathescape=true]
def randB(x) {
  r := havoc;
  return r
}
\end{lstlisting}
\end{minipage}
\end{tabular}
\vspace{-1em}

\caption{Implementations of a function which returns an integer within
  a bound.}
\label{fig:FunImplEx}
\vspace{-1em}
\end{figure}

\begin{example}
\label{ex:RandBSpecs}
To see how universal and existential specifications work together to
describe a function's behavior, consider a function \textsf{randB(x)}
which is intended to return some integer between 0 and its argument
\textsf{x}. We can write a universal specification requiring all
return values to be within the desired bound: \textsf{ax}$_\forall$
\textsf{randB}($x$) \{$0 < x$\} \{$0 \le \rho < x$\}. This does not,
however, \textit{guarantee} every value in this range is possible. To
express this requirement, we reify the choice of the random value
using an existential specification: \textsf{ax}$_\exists$
\textsf{randB}($x$) [$c$] \{$0 < x \land 0 \le c < x$\}
\{$\rho = c$\}. \autoref{fig:FunImplEx} lists a variety of possible
\textsf{randB} implementations; the first implementation is compatible
with the aforementioned universal specification and the third
definition is compatible with the existential specification, but only
the middle one satisfies both. Note how $c$ acts as a ghost variable
which constrains the choice of the random number. Thus, when reasoning
about a client of \textsf{randB}, we can select a concrete value for
$c$ that forces the desired result.
\end{example}

Equipped with a context of existential specifications
$S_\exists \in \mathcal{N} \rightarrow \mathit{FE}$ , we characterize
the set of behaviors a program \emph{must} exhibit under every
$\exists$-compatible implementation context via an underapproximate
semantics for \FunImp{} programs. The judgements of this semantics are
denoted as $S_\exists\; \vdash \sigma, p \productive \Sigma$, which
reads as: under context $S_\exists$ and initial state $\sigma$, the
program $p$ will produce at least one final state in the set of states
$\Sigma$. The evaluation rules of this semantics are given in
\autoref{fig:ProductivityRules}. Most of the rules in
\autoref{fig:ProductivityRules} adapt the \FunImp{} evaluation rules
to account for the fact that commands now produce \emph{sets} of
states from an initial state. For example, the evaluation rule for
sequences, \textsc{ESeq$_\exists$}, states that $s_2$ produces a final
state corresponding to every state in the set produced by $s_1$. The
rule for function calls, \textsc{ECall$_\exists$}, is the most
interesting: it \emph{chooses} one of the behaviors guaranteed by the
existential specification of a function and produces a \emph{set} of
final states for every return value consistent with that choice.

\begin{figure}[h]
    \begin{mathpar}
      \small
      \inferrule*[right=ESkip$_\exists$]
      {
      }
      {
        S_\exists \vdash \sigma, \sskip \productive
        \{ \sigma \}
      }

      \inferrule*[right=EHavoc$_\exists$]
      {  }
      {
        S_\exists \vdash \sigma, x := \textsf{havoc} \productive
        \{\sigma' \;|\; \exists v. [x \mapsto v]\sigma' \}
      }

      \inferrule*[right=EAssn$_\exists$]
      {
        \sigma, a \Downarrow v
      }
      {
        S_\exists \vdash \sigma, \asgn{x}{a} \productive
        \{ [x \mapsto v]\sigma \}
      }

      \inferrule*[right=EConsq$_\exists$]
      {  S_\exists \vdash \sigma, s \productive \Sigma \\
        \Sigma \subseteq \Sigma'
      }
      {
        S_\exists \vdash \sigma, s \productive \Sigma'
      }

      \inferrule*[right=ESeq$_\exists$]
      {
        S_\exists \vdash \sigma, s_1 \productive \Sigma \\
        \forall \sigma' \in \Sigma.\; S_\exists \vdash \sigma', s_2 \productive \Sigma'
      }
      {
        S_\exists \vdash \sigma, \seq{s_1}{s_2} \productive \Sigma'
      }

      \inferrule*[right=ELpT$_\exists$]
      {
        \sigma, b \Downarrow \textsf{true} \\
        S_\exists \vdash \sigma, c \productive \Sigma \\\\
        \forall \sigma' \in \Sigma.\; S_\exists \vdash \sigma', \while{b}{c} \productive \Sigma'
      }
      {
        S_\exists \vdash \sigma, \while{b}{c} \productive \Sigma'
      }

      \inferrule*[right=ELpF$_\exists$]
      {
        \sigma, b \Downarrow \textsf{false} \\
      }
      {
        S_\exists \vdash \sigma, \while{b}{c} \productive \{ \sigma \}
      }

      \inferrule*[right=EIfT$_\exists$]
      {
        \sigma, b \Downarrow \textsf{true} \\
        S_\exists \vdash \sigma, s_1 \productive \Sigma
      }
      {
        S_\exists \vdash \sigma,\ite{b}{s_1}{s_2} \productive \Sigma
      }

      \inferrule*[right=EIfF$_\exists$]
      {
        \sigma, b \Downarrow \mathtt{\bot} \\
        S_\exists \vdash \sigma, s_2 \productive \Sigma
      }
      {
        S_\exists \vdash \sigma,\ite{b}{s_1}{s_2} \productive \Sigma
      }

      \inferrule*[right=ECall$_\exists$]
      {
        S_\exists(f) = \textsf{ax$_\exists$ } f(\overline{x})\, [c]\, \{P\}\, \{ Q \} \\
        S_\exists \vdash \sigma, \overline{a} \Downarrow \overline{v} \\
        [ \overline{x} \mapsto \overline{v}] \models P[\overline{k}/\overline{c}]  \\
      }
      {
        S_\exists \vdash \sigma, y := f(\overline{a}) \productive \\
        \{ \sigma' \;|\; \exists r.\: \sigma' = [y \mapsto r]\sigma \;\land\;
        [\rho \mapsto r, \overline{x} \mapsto \overline{v}] \models Q[\overline{k}/\overline{c}]
        \}
      }
    \end{mathpar}
    \vspace{-2em}
    \caption{The existential evaluation relation.}
    \label{fig:ProductivityRules}
    \vspace{-1em}
\end{figure}

Every set of final states for a program $p$ produced by these
semantics under $S_{\exists}$ includes a possible final state of $p$
when evaluated under any $\exists$-compatible implementation context.
For this reason, we term the evaluations of $p$ using $\productive$
the \emph{underapproximate executions} of $p$ under $S_{\exists}$.
\begin{theorem}
  \label{thm:ECompat}
  If there is an underapproximate evaluation of program $p$ to a set
  of states $\Sigma$ from an initial state $\sigma$ under
  $S_{\exists}$, then $p$ must terminate in at least one final state
  $\sigma' \in \Sigma$ when it is run from $\sigma$ under an
  implementation context $I$ that is $\exists$-compatible with
  $S_{\exists}$.
\end{theorem}

\subsection{Approximating \AEH{} behaviors}

Taken together, the over- and under-approximate semantics allow us to
relate the \AEH{} behaviors of multiple client programs under every
$\forall$- and $\exists$-compatible implementation context. This
admits a modular reasoning principle, where if a set of clients can be
shown to exhibit some behaviors using the overapproximate and
underapproximate semantics, linking the client with any compatible
environment will continue to exhibit those behaviors. The key
challenge to ensuring these \AEH{} behaviors is identifying, for every
overapproximate execution, an appropriate selection of choice
variables that cause the underapproximate executions to evaluate to a
collection of final states satisfying a desired \AEH{} property.

\begin{example}
  Consider the second example from the introduction, and assume that
  \lstinlineGNI|randB| has the universal and existential
  specifications from \autoref{ex:RandBSpecs}. To ensure that
  \lstinlineGNI|encode| does not reveal anything about its secret
  input via its public output, it suffices to establish that for any
  universal execution of \lstinlineGNI|encode| on a specific input,
  every other possible input to \lstinlineGNI|encode| could produce
  the same encoded message under the existential semantics.
  The first execution begins with the statement \lstinlineGNI|int
  key$^\texttt{H}_\forall$ = randB(MAX_INT)| (for convenience, we
  annotate program variables from the first and second executions with
  the subscripts $\forall$ and $\exists$, respectively). By
  \textsc{ECall$_{\forall 1}$}, this statement will update
  \lstinlineGNI|key|$^\texttt{H}_\forall$ to hold a value between $0$
  and \texttt{MAX\_INT}. The function then encodes the message using
  this key, and returns the result. In order to show this leaks
  nothing, we need to establish a corresponding execution of
  \lstinlineGNI|encode| that returns this same result
  regardless of the value of its argument. In effect, this
  amounts to finding a strategy for instantiating the choice variable
  in \textsc{ECall$_\exists$} to assign an appropriate value to
  \lstinlineGNI|key|$^\texttt{H}_\exists$. In this case, the choice is
  straightforward: we need a $c$ such that $c$ \texttt{xor}
  \texttt{msg}$^\texttt{H}_\exists =
  $\texttt{enc}$^\texttt{L}_\forall$.  Using
  \texttt{msg}$^\texttt{H}_\exists$ \texttt{xor}
  \texttt{enc}$^\texttt{L}_\forall$ for $c$ in
  \textsc{ECall$_\exists$} achieves the desired result.  Using this
  strategy, we can construct an appropriate execution in response to
  \emph{every} execution of \lstinlineGNI|encode|. In contrast, if our
  existential specification were \textsf{ax}$_\exists$
  \textsf{randB}($x$) [ ] \{$0 < x$\} \{$0 \le \rho < x$\}, it would
  only guarantee the existence of a single result, and there would be
  no workable strategy.  Indeed, the first definition of
  \lstinlineGNI|randB| in \autoref{fig:FunImplEx} satisfies this
  specification, and \lstinlineGNI|encode| will always leak the full
  message when using this implementation!
\end{example}

%% file: sections/04-rhle.tex
\section{\rhle{}}
\label{sec:HLE}

\vspace{-.7em} We now present \rhle{}, a relational program logic for
proving that a collection of \FunImp{} programs exhibit some desired
set of \AEH{} behaviors. As a consequence of \autoref{thm:ACompat} and
\autoref{thm:ECompat}, this entails that properties established in
\rhle{} will continue to hold when the programs are linked with any
compatible implementation context.

\rhle{} specifications use \emph{relational} assertions (denoted
$\Phi, \Psi \in \mathcal{A}$) to relate the execution of multiple
programs. As normal assertions are predicates on a single state, a
relational assertion is a predicate on multiple states. Each program
in a RHLE triple operates over a distinct state space. To disambiguate
between variables that occur in multiple copies, shared variable names
are annotated with an identifier unique to each program. Following
existing convention~\cite{Sousa+CHL,Benton+RHL}, we use a natural
number to identify which state a variable belongs to. As an
example, the relational assertion $x_1 \le x_2$ is a binary predicate
over (at least) two states. This assertion is satisfied by any set of
two (or more) states where the value of $x$ in the first state is less
than or equal to the value of $x$ in the second.

\rhle{} triples have the form
$\semrhletrip{\Phi}{\overline{p_\forall}}{\overline{p_\exists}}{\Psi}$
and assert that \textit{for all} universal executions of the programs
$\overline{p_\forall}$, \textit{there exist} existential executions of
the programs $\overline{p_\exists}$ satisfying the relational
pre- and post-condition $\Phi$ and $\Psi$:
\begin{align*}
  \semrhletrip{\Phi}{\overline{p_\forall}}{\overline{p_\exists}}{\Psi}
  \equiv \quad
  &  \forall \overline{\sigma_\forall} \; \overline{\sigma_\exists} \;
  \overline{\sigma_\forall'}.\;\;
    \overline{\sigma_\forall}, \; \overline{\sigma_\exists} \models \Phi
    \:\land\:
    S_\forall \vdash \overline{\sigma_\forall}, \;
    \overline{p_\forall} \Downarrow_\forall \overline{\sigma_\forall'}
    \implies \\
  & \quad
  \exists \Sigma.\;
    S_\exists \vdash \overline{\sigma_\exists}, \; \overline{p_\exists}
    \Downarrow_\exists \Sigma\: \land\: \forall \sigma_\exists' \in \Sigma.\;
    \overline{\sigma_\forall'}, \; \overline{\sigma_\exists'} \models \Psi
\end{align*}

\begin{table}[t!]
  \begin{center}
    \begin{tabular}{ll}
      \textbf{Property}
      & \textbf{\rhle{} Assertion}
      \\\hline
      \footnotesize Refinement
      & \small
        $S_\forall,\: S_\exists \models
        \left\langle {\overline{x_1} = \overline{x_2}} \right\rangle
        \asgn{y_1}{f(\overline{x_1})} \esim \asgn{y_2}{f(\overline{x_2})}
        \left\langle {y_1 = y_2} \right\rangle$
      \\\hline
      \footnotesize Noninterference
      & \small
        $S_\forall,\: S_\exists \models
        \left\langle {low_1 = low_2} \right\rangle
        p_1 \esim p_2
        \left\langle {low_1 = low_2} \right\rangle$
      \\\hline
      \footnotesize Injectivity
      & \small
        $S_\forall,\: S_\exists \models
        \left\langle {x_1 \neq x_2} \right\rangle
        {\asgn{y_1}{f(x_1)} \circledast \asgn{y_2}{f(x_2)}}
        \esim {\sskip}
        \left\langle {y_1 \neq y_2} \right\rangle$
      \\\hline
      \footnotesize Nondeterminism
      & \small
        $S_\forall,\: S_\exists \models
        \left\langle {x_1 = x_2} \right\rangle
        {\sskip}
        \esim {\asgn{y_1}{f(x_1)} \circledast \asgn{y_2}{f(x_2)}}
        \left\langle {y_1 \neq y_2} \right\rangle$
      \\\hline
    \end{tabular}
  \end{center}
  \caption{Example \rhle{} assertions. In the second row, $low_x$
    refers to the low security state in program $p_x$; note the \AEH{}
    relationship must hold for \emph{any} pair of initial high
    security values, so $high_x$ is not constrained in the
    precondition.}
  \label{fig:RHLEAssertEx}
  \vspace{-2em}
\end{table}

\noindent We use $\circledast$ to delineate different programs on the
universal and existential sides of $\esim$ so that, e.g., a sequence
of $n$ programs $\overline{p}$ is also denoted as
$p_1 \circledast \ldots \circledast p_n $. For example, to assert the
program $\asgn{x}{\textsf{havoc}}$ is nondeterministic, we write a
RHLE triple with two copies of the program, adding a subscript to the
variable $x$ in each for clarity:
$\mathord{\boldsymbol{\cdot}} \models \left\langle {\top}
\right\rangle {\sskip} \esim
{\asgn{x_1}{\textsf{havoc}}\circledast\asgn{x_2}{\textsf{havoc}}}
\left\langle {x_1 \neq x_2} \right\rangle$.  This triple says that,
for all starting states and all executions of the trivial program
$\sskip$, there exist executions of the programs
$\asgn{x_1}{\textsf{havoc}}$ and $\asgn{x_2}{\textsf{havoc}}$ such
that $x_1 \neq x_2$ after both programs have executed. Note that
$\circledast$ is \emph{not} a concatenation operator; it does nothing
more than delineate multiple programs in a RHLE triple.
\autoref{fig:RHLEAssertEx} gives some additional examples of \rhle{}
assertions.

\begin{figure}
  \vspace{-1em}
  \begin{minipage}{.49\textwidth}
    \begin{mathpar}
      \small
      \inferrule*[right=Finish]
    {
    }
    {
      \rhletrip{\Phi}
          {\overline{\textrm{skip}}}
          {\overline{\textrm{skip}}}
          {\Phi}
        }
      \end{mathpar}
    \end{minipage}
    \begin{minipage}{.49\textwidth}
      \begin{mathpar}
        \small
        \inferrule*[right=SkipI]
        {
          \rhletrip{\Phi}
          {\overline{\seq{p_\forall}{\sskip}}}
          {\overline{\seq{p_\exists}{\sskip}}}
          {\Psi}
        }
        {
          \rhletrip{\Phi}
          {\overline{p_\forall}}
          {\overline{p_\exists}}
          {\Psi}
        }
      \end{mathpar}
    \end{minipage}
  \begin{mathpar}
    \small
    \inferrule*[right=Step$\forall$]
    {
      \forall \overline{\sigma}\; \overline{\sigma_\exists}.\:
      \hltrip{\Phi\: |_i\: \overline{\sigma},\; \overline{\sigma_\exists}}
      {s_i}
      {\Phi'\: |_i\: \overline{\sigma},\; \overline{\sigma_\exists}} \\
      \rhletrip{\Phi'}
      {p_1 \circledast \ldots \circledast s_i' \circledast \ldots \circledast s_n}
      {\overline{p_\exists}}
      {\Psi}
    }
    {
      \rhletrip{\Phi}
      {p_1 \circledast \ldots \circledast \seq{s_i}{s_i'} \circledast \ldots \circledast p_n}
      {\overline{p_\exists}}
      {\Psi}
    }

    \inferrule*[right=Step$\exists$]
    {
      \forall \overline{\sigma_\forall}\; \overline{\sigma}. \:
      \hletrip
      {\Phi\: |_i\: \overline{\sigma_\forall},\; \overline{\sigma}}
      {s_i}
      {\Phi'\: |_i\: \overline{\sigma_\forall},\; \overline{\sigma}} \\
      \rhletrip{\Phi'}
      {\overline{p_\forall}}
      {p_1 \circledast \ldots \circledast s_i' \circledast \ldots \circledast p_n}
      {\Psi}
    }
    {
      \rhletrip{\Phi}
      {\overline{p_\forall}}
      {p_1 \circledast \ldots \circledast \seq{s_i}{s_i'} \circledast \ldots \circledast p_n}
      {\Psi}
    }
  \end{mathpar}
  \vspace{-2em}
  \caption{Core \rhle{} proof rules.}
  \label{fig:RHLERules}
  \vspace{-1em}
\end{figure}

The core logic of \rhle{} is given in
\autoref{fig:RHLERules}. Relational proofs are built by reasoning
about the topmost statement of either one of the universally
quantified programs via the \textsc{Step}$\forall$ rule or one of the
existentially quantified programs using the \textsc{Step}$\exists$
rule. Once all program statements have been considered, final proof
obligations can be discharged using the \textsc{Finish} rule. The
\textsc{SkipI} rule is used to ensure that all programs end with
$\sskip$, so that \textsc{Finish} can be applied. Both \textsc{Step}
rules rely on non-relational logics for reasoning about the universal
$\hltrip{P}{p}{Q}$ and existential $\hletrip{P}{p}{Q}$ behaviors of
single statements; we will present the details of both logics
shortly. The \textsc{Step} rules employ a projection operation,
$\overline{\sigma} |_i \Psi$, which maps a relational assertion to a
non-relational one. Given a collection of $n$ states,
$\Psi\: |_i\: \overline{\sigma}$ is satisfied by any state $\sigma'$
which satisfies $\Psi$ when inserted at the $i$th position:
\[ \sigma' \models\ \Psi\: |_i\: \overline{\sigma} ~~\equiv \quad
\sigma_1,\ldots, \sigma_{i-1}, \sigma', \sigma_{i+1}, \ldots, \sigma_n \models
\Psi
\]
\noindent In effect, this operation ensures the states of the other programs
remain unchanged when reasoning about the $i$th program in the triple.

\paragraph{Universal Hoare Logic}
The program logic for universal executions has a standard partial
correctness semantics:
\begin{align*}
  \semhltrip{P}{p}{Q} ~~\equiv \quad &
  \forall \sigma, \sigma'.\;  \sigma \models P
  \wedge S_\forall \vdash \sigma, p \Downarrow_\forall \sigma'
  \implies \sigma' \models Q
\end{align*}
\noindent The rules of this logic are largely
standard\footnote{\apref{sec:HL} gives a full
  listing of the rules of this logic.}, except for the rule for
function calls, which uses a context of universal function
specifications:
\begin{mathpar}
  \inferrule*[right=$\forall$Spec]
  {
    S_{\forall}(f)= \textsf{ax}_{\forall}\; f (\overline{x}) \{P\} \{Q\}
  }
  {
    \hltrip{
      \begin{aligned}
      & P [\overline{a}/\overline{x}] ~~~\land \\
      \forall v. & Q [v/\rho; \overline{a}/\overline{x}]
      & \implies R[v/y]
      \end{aligned}
    }
    {\asgn{y}{f(\overline{a})}}{R}
  }
\end{mathpar}

\paragraph{Existential Hoare Logic}
\label{sec:EHL}
The assertions of our program logic for existential executions say
that, for any state meeting the precondition, there \textit{exists} an
execution of the program ending in a set of states meeting the
post-condition:
\begin{align*}
 \semhletrip{P}{p}{Q} ~~\equiv \quad &
 \forall \sigma.\; \sigma \models P \implies
\exists \Sigma.\; S_\exists \vdash \sigma, p \Downarrow_\exists \Sigma
~~~\land~~~
\forall \sigma' \in \Sigma.\; \sigma' \models Q
\end{align*}
\noindent These rules are largely standard \textit{total} Hoare logic
rules\footnote{The full existential logic is presented in
  \apref{sec:FullExistentialRules}.}, augmented with
a rule for calls to existentially specified functions:
\begin{mathpar}
\inferrule*[right=$\exists$Spec]
{
  S_\exists(f) = \textsf{ax}_\exists f(\overline{x})\; [\overline{c}]\; \{P\}\; \{Q\}
}
{
      \hletrip{
        \begin{minipage}{.34\columnwidth}
          \begin{align*}
            \exists \overline{k}.
            \;(
            &[\overline{x} \mapsto \overline{a}] \models P[\overline{k}/\overline{c}] \\
            ~~ \land ~~ & \exists v.
                          [\rho \mapsto v, \overline{x} \mapsto
                          \overline{a}] \models Q[\overline{k}/\overline{c}]
            \\
            ~~\land~~ &  \forall v.
                        [\rho \mapsto v, \overline{x} \mapsto
                        \overline{a}] \models
                        Q[\overline{k}/\overline{c}] \\
            & \quad\quad \implies
                        R[v/y]
                        )
          \end{align*}
        \end{minipage}
      }
      {\asgn{y}{f(\overline{a})}}{R}
}
\end{mathpar}

The precondition of this rule is quantified over instantiations
$\overline{k}$ of the specification's choice variables. The first of
the three conjuncts under this quantifier ensures that the statement
is executed in a state satisfying the function's precondition. The
next conjunct ensures that the function's post-condition is
inhabited. The final conjunct requires that every possible return
value satisfying the function's post-condition also satisfies the
triple's post-condition.

\begin{example}
  Given the existential specification \textsf{ax}$_\exists$
  \textsf{zeroOrOne}() [$c$] \{$c = 0 \lor c = 1$\} \{$\rho = c$\}, we
  can use \textsc{$\exists$Spec} (along with the rule for while loops
  given in \apref{sec:FullExistentialRules}) to prove the existential assertion
  $\hletrip{k=0}{\while{k <
      4}{\asgn{k}{k + \textsf{zeroOrOne()}}}}{k = 4}$. This loop
  \emph{could} loop forever by choosing to add $0$ to $k$ at every
  iteration. Nevertheless, by using measure $4-k$ with the
  well-founded relation $<$ and instantiating the choice variable with
  $1$ at each iteration, we can prove a terminating path through the
  program exists.
\end{example}

\subsection{Synchronous Rules}
\label{sec:sync+rules}
While the rules in \autoref{fig:RHLERules} are sufficient to reason
about relational properties, it is possible to lessen the verification
burden for structurally similar programs by employing
\emph{synchronous rules} which exploit structural similarities between
the programs being verified~\cite{Nagasamudram2021}. Reasoning over
similar control flow structures in lockstep can reduce the space of
states verification must consider and simplify loop invariants. This
is particularly useful when reasoning about
\emph{hyperproperties}~\cite{Clarkson+HYP}, or relational properties
on multiple executions of the \emph{same} program. In order to more
easily reason about structurally similar programs, \rhle{} also
includes synchronous rules inspired by the Cartesian loop logic
presented by Sousa and Dillig~\cite{Sousa+CHL}.
\apref{sec:SynchronousRules} includes a full listing of these rules.

\begin{example}
  \label{ex:SyncLoops}
  Consider proving that
  $\textsf{\while{(x < 10)}{\asgn{y}{y + \textsf{randB}(9)}}}$ refines
  $\textsf{\while{(x < 10)}{\asgn{y}{y + \textsf{randB}(5)};
      \asgn{y}{y + \textsf{randB}(6)}}}$. Intuitively, the first
  program refines the second because the bodies of the loops are
  themselves refinements. A proof using only the rules in
  \autoref{fig:RHLERules} is unable to take advantage of this
  intuition, however. Instead, the proof requires a sufficiently
  strong invariant characterizing the behavior of the entire loop on
  the left, and then an invariant for the righthand program that
  accounts for the behavior of individual iterations of the lefthand
  loop.

The \textsc{SyncLoops} rule is designed for this situation:
  \begin{mathpar}
    \inferrule*[right=SyncLoops]
    {
      \rhletrip{\mathbb{I} \land \bigwedge_{0 \le i \le n} b_i}
               {s_0 \circledast \cdots \circledast s_k}{s_{k+1}
                 \circledast \cdots \circledast s_n}
               {\mathbb{I}} \\
      \mathbb{I} \land \bigwedge_{0 \le i \le n} \neg b_i \implies \Psi \\
      \mathbb{I} \land \neg \bigwedge_{0 \le i \le n} b_i \implies \bigwedge_{0 \le i \le n} \neg b_i
    }
    {
      S_\forall, S_\exists \vdash \epre{\mathbb{I}}
      \while{b_0}{s_0} \circledast \cdots \circledast \while{b_k}{s_k} \\
      \phantom{S_\forall, S_\exists \vdash} \esim
      \while{b_{k+1}}{s_{k+1}} \circledast \cdots \circledast \while{b_n}{s_n}
      \epost{\Psi}
    }
  \end{mathpar}
  The first premise of this rule says that executing all loop bodies
  preserves some invariant $\mathbb{I}$, the second ensures the
  invariant is strong enough to imply the postcondition, and the third
  requires all loops to end on the same iteration. Since this
  invariant is reestablished after the execution of every loop body;
  the invariant that $\textsf{y}_1$ and $\textsf{y}_2$ are equal at
  each iteration suffices to verify this example.
\end{example}

\subsection{Soundness}
The combination of the core and synchronous rules provide a
sound methodology for reasoning about \AEH{} properties:
\begin{theorem}[RHLE is Sound]
  Suppose
  $\rhletrip{\Phi}{\overline{p_\forall}}{\overline{p_\exists}}{\Psi}$. Then,
  for any function context $I$ compatible with $S_\forall$ and
  $S_\exists$, any set of initial states $\overline{\sigma_\forall}$
  and $\overline{\sigma_\exists}$ satisfying $\Phi$, and every
  collection of final states $\overline{\sigma_\forall'}$ of
  $\overline{p_\forall}$, there must exist a collection of final
  states produced by $\overline{p_\exists}$ that, together with
  $\overline{\sigma_\forall'}$, satisfies the relational post-condition
  $\Psi$.
\end{theorem}

%% file: sections/05-verification.tex
\section{Verification}
\label{sec:Verification}

\begin{wrapfigure}{l}{0.54\textwidth}
  \vspace{-2.1em}
  \begin{minipage}{0.54\textwidth}
\begin{algorithm}[H]
  \SetKwProg{Match}{match}{\string:}{}
  \SetKwInOut{Params}{Inputs}
  \SetKwInOut{Output}{Output}
  \SetKwFunction{WPGen}{VCGen}
  \SetKwFunction{Check}{Verify}
  \DontPrintSemicolon
  \Params{$\Phi$, relational precondition \\
          $p_{\forall}$, universal programs \\
          $p_{\exists}$, existential programs \\
          $\Psi$, relational postcondition}
  \Output{$\rhletripnosigma{\Phi}{p_\forall}{p_\exists}{\Psi}$ validity}
\Begin{
  $\overline{\Psi} \leftarrow (\varnothing, \varnothing, \Psi)$ \;
  $(\overline{a}, \overline{e}, \Psi') \leftarrow$
  \WPGen{$\overline{\sskip{}; p_\forall}, \overline{\sskip{}; p_\exists}, \overline{\Psi}$} \;
  \Return \Check{$\forall \overline{a} \exists \overline{e}.\; \Phi \implies \Psi'$}
}
\caption{RHLEVerify}
\label{alg:RHLEVerify}
\end{algorithm}
\end{minipage}
\vspace{-2em}
\end{wrapfigure}

We now turn to the relational verification algorithm based on \rhle{},
presented in \autoref{alg:RHLEVerify}. The algorithm is implicitly
parameterized over a pair of universal and existential contexts, and
\texttt{Verify}, a decision procedure for checking validity of a
formula in the underlying assertion logic. The bulk of the work is
delegated to \texttt{VCGen}, presented in \autoref{alg:VCGen}, which
builds a weakest relational precondition for the input \rhle{}
triple. The algorithm then checks that the \rhle{} triple's
precondition entails the calculated weakest precondition.

\begin{figure}[t!]
\begin{algorithm}[H]
  \SetKwProg{Match}{match}{\string:}{}
  \SetKwInOut{Params}{Inputs}
  \SetKwInOut{Output}{Output}
  \SetKwFunction{recurse}{VCGen}
  \SetKwFunction{Check}{Verify}
  \SetKwFunction{FindInvariant}{FindInvariant}
  \SetKwFunction{FreeVars}{FreeVars}
  \SetKwFunction{Freshen}{Freshen}
  \DontPrintSemicolon
  \Params{$p_{\forall}$, a set of universal programs \\
          $p_{\exists}$, a set of existential programs \\
          ${\overline{\Psi}} = (Q_\forall, Q_\exists, \Psi)$, $\Psi$ a postcondition
            with quantified variables $Q_\forall$, $Q_\exists$
          }
  \Output{$(\{v_0, \ldots, v_n\}, \{w_0, \ldots, w_n\}, \Phi)$ such that
    $v_i$, $w_i$ free in $\Phi$ and
    $\rhletripnosigma{\Phi}{p_{\forall}}{p_{\exists}}{\Psi}$ is valid if $\forall
    v_0, \ldots, v_n\ \exists w_0, \ldots w_n.\ \Phi \implies \Psi$ }
\Begin{
  \Match{$p_\forall \esim p_\exists$}{
    \Case{$\overline{\sskip} \esim \overline{\sskip}$}{
      \Return{$\overline{\Psi}$}
    }
    \Case{$\overline{p'_\forall} \circledast (s_1;s_2) \circledast
      \overline{p''_\forall} \esim p_\exists$ \KwSty{where} $s_2$ not
      a loop
    }{
      \recurse{$\overline{p'_\forall} \circledast s_1 \circledast
      \overline{p''_\forall}, p_\exists, \texttt{VC}_\forall(s_2, \overline{\Psi})$}
    }
    \Case{$p_\forall \esim \overline{p'_\exists} \circledast (s_1;s_2)
      \circledast \overline{p''_\exists}$ \KwSty{where} $s_2$ not a
      loop
    }{
      \recurse{$p_\forall, \overline{p'_\exists} \circledast s_1
        \circledast \overline{p''_\exists} , \texttt{VC}_\exists(s_2, \overline{\Psi})$}
    }
    \Case{$\overline{p'_\forall} \circledast s_1;\ite{b}{s_t}{s_e} \circledast \overline{p''_\forall}
      \esim p'_\exists$} {
      $(Q_\forall, Q_\exists, \Psi_T) \leftarrow$ \recurse{$\overline{p'_\forall} \circledast
        s_1; s_t \circledast
        \overline{p''_\forall}, p_\exists, b\implies \Psi)$} \;
      $(Q'_\forall, Q'_\exists, \Psi_E) \leftarrow$ \recurse{$\overline{p'_\forall} \circledast
        s_1; s_e \circledast
        \overline{p''_\forall}, p_\exists, \lnot b\implies \Psi)$} \;
        \Return{($Q_\forall \cup Q'_\forall, Q_\exists \cup Q_\exists, \Psi_T \land \Psi_E$)}
    }
\Case{$p_{0};\whilenoend{b_0}{s_0} \circledast
      \cdots \circledast p_{i-1};\whilenoend{b_{i-1}}{s_{i-1}}
          \esim p'_{i};\whilenoend{b_{i}}{s_{i}} \circledast \cdots \circledast p'_{n};\whilenoend{b_n}{s_n}$}{
          $\mathbb{I} \leftarrow $
          \FindInvariant{$\whilenoend{b_0}{s_0} \circledast
      \cdots \circledast \whilenoend{b_{i-1}}{s_{i-1}}
          \esim \whilenoend{b_{i}}{s_{i}} \circledast \cdots \circledast \whilenoend{b_n}{s_n}$} \;
      $(Q'_\forall, Q'_\exists,\Psi_{body}) \leftarrow$
         \recurse{$s_0 \circledast
             \cdots \circledast s_{i-1} \esim s_i \circledast
             \cdots \circledast s_n, \mathbb{I}$} \;
      $inductive \leftarrow \mathbb{I} \land \bigwedge_{0 \le i \le n} b_i \implies \Psi_{body}$ \;
      $lockstep \leftarrow \mathbb{I} \land \neg \bigwedge_{0 \le i \le n} b_i
          \implies \bigwedge_{0 \le i \le n} \neg b_i$ \;
      $post \leftarrow \mathbb{I} \land \bigwedge_{0 \le i \le n} \neg b_i \implies \Psi$ \;
      $(Q_\forall, Q_\exists, \Psi) \leftarrow \overline{\Psi}$ \;
      \If{\Check{$Q_\forall \cup Q'_\forall,\; Q_\exists \cup Q'_\exists,\; inductive \land lockstep \land post$}}{
        \recurse{$\overline{p}, \overline{p'}, (Q_\forall, Q_\exists, \mathbb{I})$}
      } \Else {
        \KwSty{next case}
      }
    }
    \Case{$\overline{p'_\forall} \circledast (s_1;s_2) \circledast
      \overline{p''_\forall} \esim p_\exists$}{
      \recurse{$\overline{p'_\forall} \circledast s_1 \circledast
      \overline{p''_\forall}, p_\exists, vc_\forall(s_2, \overline{\Psi})$}
    }
    \Case{$p_\forall \esim \overline{p'_\exists} \circledast (s_1;s_2)
      \circledast \overline{p''_\exists}$}{
      \recurse{$p_\forall, \overline{p'_\exists} \circledast s_1
        \circledast \overline{p''_\exists} , vc_\exists(s_2, \overline{\Psi})$}
    }
  }
}
\caption{VCGen}
\label{alg:VCGen}
\end{algorithm}
\vspace{-2.5em}
\end{figure}

The body of \texttt{VCGen} builds a formula by recursively generating
verification conditions for the input programs statement by
statement. This loop tries to maximize opportunities to apply
synchronous rules at each step, as these rules allow us to
simultaneously generate proof obligations for multiple subprograms, as
discussed in \autoref{sec:sync+rules}. After establishing there are
still program statements to step over (lines 3--4), \texttt{VCGen}
looks for and processes any trailing program statements which are not
loops (lines 5--8), as such statements are not subject to synchronous
rule applications. To process individual program statements,
\texttt{VCGen} relies on a pair of verification condition generators,
\texttt{VC}$_\forall$ and \texttt{VC}$_\exists$, for the
non-relational program logics. These functions are largely standard
weakest precondition generators extended with support for existential
function calls. The consequents of \textsc{$\forall$Spec} and
\textsc{$\exists$Spec} immediately yield weakest precondition rules,
so that if
$S_{\forall}(f)= \textsf{ax}_{\forall}\; f (\overline{x}) \{P\} \{Q\}$
and
$S_\exists(f) = \textsf{ax}_\exists f(\overline{x})\; [\overline{c}]\;
\{P\}\; \{Q\}$, then:
\begin{align*}
 \texttt{VC}_\forall(\Psi, \asgn{y}{f(\overline{a})}) &=
P [\overline{a}/\overline{x}] \land \forall v. Q [v/\rho;
\overline{a}/\overline{x}] \implies \Psi[v/y] \\
\texttt{VC}_\exists(\Psi, \asgn{y}{f(\overline{a})}) & =
\exists \overline{k}.\;(
   [\overline{x} \mapsto \overline{a}] \models P[\overline{k}/\overline{c}]
           ~ \land ~  \exists v.
                          [\rho \mapsto v, \overline{x} \mapsto
                          \overline{a}] \models
                                                              Q[\overline{k}/\overline{c}]
  \\
  & \phantom{\quad \exists \overline{k}.~~}
            ~\land~  \forall v.
                        [\rho \mapsto v, \overline{x} \mapsto
                        \overline{a}] \models
                        Q[\overline{k}/\overline{c}]
\implies
                        \Psi[v/y]
                        )
\end{align*}
\noindent If the first three cases fail, the final statements of all
the remaining programs are loops. In this case, \texttt{VCGen}
attempts to simultaneously process the loops (lines 9--19) \`{a} la
the \textsc{SyncLoops} rule in \autoref{ex:SyncLoops}. To be
eligible for fusion, loops must execute in lockstep.  This condition
is checked (line 16) before returning; if loops may execute different
numbers of times, the algorithm proceeds to the next match case. If no
synchronized reasoning is possible, \texttt{VCGen} defaults to
stepping over an arbitrary loop in one of the programs (lines 20--23).

\texttt{VCGen} is parameterized over a procedure called
\texttt{FindInvariant}, which acts as an oracle for relational loop
invariants. Our prototype implementation of \autoref{alg:RHLEVerify}
currently requires loops to be annotated with their invariants; these
annotations are used to implement \texttt{FindInvariant}. We have
experimented with adapting both purely
logical~\cite{Houdini,DilligAbductiveInv} and data-driven
approaches~\cite{padhi2016data,LoopInvGen} for invariant inference,
but have yet to discover one that is effective for our larger
benchmarks. Unlike traditional loop invariants, which must be
re-established on every possible execution of the loop body,
invariants in existentially quantified executions need only be
re-established on a subset of the possible executions of the body. A
robust invariant inference approach thus requires finding not only the
invariant itself, but a strategy for instantiating choice variables
that consistently re-establish the chosen invariant. Scalable
invariant inference for existentially quantified executions is an
important and interesting direction for future work.

\apref{sec:VerificationExample} includes an example application of
\autoref{alg:RHLEVerify} to \texttt{RandB}.

%% file: sections/06-evaluation.tex
\section{Implementation and Evaluation}
\label{sec:Evaluation}
\vspace{-.7em} To evaluate our approach, we have implemented \orhle{},
a publicly available~\cite{OrhleZenodo} automatic program verifier
based on \autoref{alg:RHLEVerify}. \orhle{} is implemented in Haskell,
and uses Z3 as a backend solver to fill the role of \texttt{Verify}.
As previously mentioned, invariants are provided by the programmer via
annotations in the code. Input to \orhle{} consists of a collection of
\FunImp{} programs, a declaration of how many copies of each program
should be included in the universal and existential contexts, and a
collection of function specifications expressed using the SMT-LIB2
format. Functions can have both universal and existential
specifications, with the latter containing declarations of choice
variables. \apref{sec:orhleInput} has example \orhle{} input listings.
\orhle{} outputs a set of verification conditions along with a success
or failure message. When a property fails to verify, \orhle{} outputs
a falsifying model.

Our evaluation addresses the following questions:

\begin{enumerate}
\item[(R1)] Is \rhle{} \emph{expressive} enough to represent a variety
  of interesting properties?
\item[(R2)] Is our approach \emph{effective}; that is, can it be used
  to verify or invalidate relational assertions about a diverse corpus
  of programs?
\item[(R3)] Is it possible to realize an \emph{efficient} implementation
  of our verification approach which returns results within a reasonable
  time frame?
\end{enumerate}

To answer these questions, we have developed a suite of $41$ programs
over $5$ kinds of relational specifications drawn from the literature.
We have also compiled an additional set of $12$ benchmarks over two
non-relational existentially quantified properties in order to
evaluate similar questions about the non-relational existential logic
from \autoref{sec:EHL}. Both sets of benchmarks contain a mix of valid
and invalid properties. We have made these benchmarks publicly
available\footnote{Branching time property benchmarks are adapted from
a proprietary source, and are thus omitted from the publicly available
benchmarks.} via GitHub~\cite{OrhleBenchmarks}.

Our benchmarks for the non-relational existential logic from
\autoref{sec:EHL} fall into two categories:

\parheading{Winning Strategy} Programs in this category play a
simplified version of the card game twenty-one. Players start with two
cards valued between 1 and 10, and can then request any number of
additional cards. The goal is to get a hand value as close to 21 as
possible without going over. The property of interest is whether an
algorithmic strategy for this game permits the \emph{possibility} of
achieving the maximum hand value of 21 given any starting hand.

\parheading{Branching Time Properties} Our next set of benchmarks are
taken from work by Cook and Koskinen~\cite{Cook2013} which considered
verification of properties of single programs expressed in CTL. The
programs in this category are adaptations of the subset of those
benchmarks which assert the existence of desirable final states and
are thus expressible in \rhle{}.

Our set of relational benchmarks cover program refinement in addition
to:

\parheading{Noninterference} Generalized noninterference is a
possibilistic information security property which ensures that
programs do not leak knowledge about high-security state via
low-security outputs. Our formalization of this property is based on
Mclean~\cite{Mclean1996} and requires that, for any execution of a program
$p$ whose state is divided into high security $p_H$ and low security
$p_L$ partitions, any other starting state with the same initial low
partition can potentially yield the same final low partition,
regardless of the high partition.

\parheading{Delimited Release} Delimited release is a relaxation of
generalized noninterference which allows for limited information about
secure state to be released. For example, given a confidential list of
employee salaries, it may be acceptable to publicize the average
salary as long as no other salary information is leaked. We formulate
delimited release as a noninterference property with an additional
condition requiring that the programs agree on the values of the
released information. For the previous example, we would add a
precondition asserting the average salary across all executions is
equal.

\parheading{Parameter Usage} Our parameter usage benchmarks check
whether a function parameter is semantically unused, in that the
existence of the parameter does not affect the program's reachable
final states. For example, the \texttt{flag} parameter in
\texttt{f(flag) = if flag then return 1 else return 1} is
syntactically used in \texttt{f}, even affecting its control flow, but
does not have any effect on \texttt{f}'s possible outputs; we
therefore consider \texttt{flag} to be semantically unused. For an
n-ary function $f(p_1, \ldots, p_n)$, we say parameter $p_i$ is
semantically unused if
\[\rhletripnosigma{v_i \neq w_i \wedge
\bigwedge_{j \neq i} v_j = w_j}{\asgn{a}{f(v_1, \ldots,
v_n)}}{\asgn{b}{f(w_1, \ldots, w_n)}}{a = b}
\]

\parheading{Flaky Tests} Tests of program behavior which can
nondeterministically pass or fail pose a significant hazard as they
can trigger false alarms or allow regressions to go undetected. We
modeled representative nondeterministic tests in \FunImp{} based on
examples from The Illinois Dataset of Flaky Tests
(IDoFT)\cite{Shi2016,Lam2019}, framing flakiness as a \AEH{} property
containing only existential executions. We consider a test verifiably
flaky when there exists both a test execution that succeeds and one
that fails. We model nondeterminsitic system behavior (e.g.,
\texttt{getCurrentTimeMs()} or the results of network calls) as
function calls. For example, to model the imprecision of thread
\texttt{sleep}s, we give the verifier leeway to sleep within a
$\pm 20$ ms window around the requested interval:
\textsf{ax}$_\exists$ \textsf{sleep}(interval, currentTime)
[sleepTime]
\{$0 \le \textrm{sleepTime} \;\land\; \textrm{interval} - 20 \le
\textrm{sleepTime} \le \textrm{interval} + 20$\}
\{$\rho = \textrm{currentTime} + \textrm{sleepTime}$\}.

\begin{figure}[t]
\footnotesize
\begin{center}
\setlength{\tabcolsep}{3pt}
\begin{tabular}{lcccccc}
\textbf{Property} & \textbf{Shape} & \textbf{Pos} & \textbf{Neg}
  & \textbf{Unk} & \textbf{Med}(ms) & \textbf{Max}(ms) \\
  \hline
Delimited Release & $\forall p_1 \exists p_2$ & 7 & 6 & 0 & 222 & 253 \\
Flaky Tests & $\exists p_1 p_2$ & 2 & 0 & 0 & 231 & 245 \\
Generalized Noninterference & $\forall p_1 \exists p_2$ & 4 & 6 & 0 & 222 & 229 \\
Parameter Usage & $\forall p_1 \exists p_2$ & 4 & 3 & 0 & 220 & 245 \\
Program Refinement & $\forall p_1 \exists p_2$ & 4 & 4 & 1 & 224 & 1367 \\
\hhline{=======}
Winning Strategy & $\exists p$ & 1 & 2 & 0 & 228 & 230 \\
Branching Time & $\exists p$ & 7 & 2 & 0 & 226 & 259 \\
\end{tabular}
\end{center}
\vspace{-1em}
\caption{\orhle{} verification results over a set of relational and
non-relational properties. The \textbf{Shape} column gives the
execution quantification pattern for the property; each property is of
the form $\forall p_0 \ldots p_n \exists q_o \ldots q_n$, where
$p_i$'s and $q_i$'s are (possibly empty) sets of executions. The
\textbf{Pos} and \textbf{Neg} columns give the number of benchmarks
over which the property holds or does not hold, respectively. The
\textbf{Unk} column gives the number of benchmarks whose verification
conditions could not be decided by the SMT solver. The \textbf{Med}
and \textbf{Max} columns give (respectively) the median and maximum
verification times in milliseconds over each set of benchmarks. }
\label{fig:eval}
\vspace{-1em}
\end{figure}

The variety of properties we were able to represent in \orhle{}
provides evidence that it is sufficiently expressive (R1). To show
that \orhle{} is both effective and efficient (R2)-(R3), we have used
it to verify and/or invalidate examples of the benchmark properties
described above. All of these experiments were done using an Intel
Core i7-6700K CPU with 8 4GHz cores. \autoref{fig:eval} presents the
results of these experiments. \orhle{} yielded the expected
verification result in all cases except for one refinement benchmark,
where the backing SMT solver (Z3) was unable to determine the validity
of the verification conditions. While most benchmarks' verification
conditions fell within the theory of linear integer arithmetic,
verification conditions fell in a non-decidable fragment of arithmetic
in this benchmark. This undecidable instance accounts for the outlier
maximum verification time in the refinement benchmarks. Overall, these
results offer evidence that \orhle{} is both effective and efficient
for verifying a variety of existential and \AEH{} properties.

%% file: sections/07-related-work.tex
\section{Related Work}
\label{sec:RelatedWork}
\vspace{-.7em}
\noindent\textit{Relational Program Logics}
Relational program logics are a common approach to verifying
relational specifications. Relational Hoare Logic~\cite{Benton+RHL}
(RHL) was one of the first examples of these logics, and is capable of
proving 2-safety properties. Relational Higher-order
Logic~\cite{Aguirre+RHOL} is a higher-order relational logic for
reasoning about higher-order functional programs expressed in a
simply-typed $\lambda$-calculus. Probabilistic RHL~\cite{Barthe+pRHL}
is a logic for reasoning about probabilistic programs in order to
prove security properties of cryptographic schemes. The relational
logic closest to \rhle{} is Cartesian Hoare Logic~\cite{Sousa+CHL}
(CHL) developed by Sousa and Dillig. This logic which provides an
axiomatic system for reasoning about $k$-safety hyperproperties along
with an automatic verification algorithm. \rhle{} can be thought of as
an extension of CHL for reasoning about the more general class of
\AEH{} properties. Nagasamudram and Naumann~\cite{Nagasamudram2021}
examine \textit{alignment completeness} for relational Hoare logics,
which classifies the ability of these logics to reason about programs
in lockstep. Banerjee et al.~\cite{Banerjee2019} introduce a relational
Hoare logic capable of reasoning about encapsulation and invariant
hiding, but which is confined to 2-safety properties.

\noindent\textit{Underapproximate Program Logics}
Several program logics have been proposed to reason about the
existence of particular executions of a single program, similar to the
non-relational existential logic presented in \autoref{sec:HLE}.
Reverse Hoare Logic~\cite{Vries+ReverseHL} is a program logic for
reasoning about reachability over single executions of programs which
have access to a nondeterministic binary choice ($\sqcup$) operator.
Incorrectness Logic~\cite{OHearn+Incorrectness} is a recent adaptation
of Reverse Hoare Logic to a more realistic programming language. While
these logics express the existence of a satisfying start state for all
satisfying end states ($\forall \sigma' \exists \sigma$), the
existential logic presented in \autoref{sec:HLE} requires there to
exist a satisfying end state for all satisfying start states
($\forall \sigma \exists \sigma'$). Reverse Hoare Logic and
Incorrectness Logic both reason about reachability over single
executions, but properties in these logics are pure
underapproximations: every state in a given postcondition must be
reachable. In contrast, our reasoning over existential specifications
is underapproximate \emph{with respect to the choice variables only}.
While every valid choice value must correspond to a reachable set of
final states, each of these sets are still overapproximate. This
feature of our existential specifications enables a natural
integration with standard Hoare logics.

First-order dynamic logic~\cite{Pratt1976} is a reinterpretation of
Hoare logic in first-order, multi-modal logic. For a program $p$, the
modal operators $[p]$ and $\langle p \rangle$ capture universal and
existential quantification over program executions. Our universal
Hoare triple $\vdash \{P\} p \{Q\}$ corresponds to $P \implies [p]Q$,
and our existential Hoare triple $\vdash [P] p [Q]_{\exists}$
corresponds to $P \implies \langle p \rangle Q$. In contrast to
\rhle{}, dynamic logic reasons about properties of single program
executions.

\noindent\textit{Prophecy Variables}
Prophecy variables were originally introduced by Abadi and Lamport
\cite{Abadi+Refinement} in order to establish refinement mappings
between state machines. Choice variables in our existential
specifications are similar to prophecy variables in that they capture
the required value of some ``future'' state, although we use them as
part of a program logic rather than to reason about refinement
mappings between state machines. Jung et al. \cite{Jung+Future}
incorporate prophecy variables into a separation Hoare logic to reason
about nondeterminism in concurrent programs, but differ from our
approach in that the program logic operates in a non-relational
setting and is designed for interactive and not automated
verification.

\noindent\textit{Relational Verification}
The concept of a hyperproperty was originally introduced by Clarkson
and Schneider~\cite{Clarkson+HYP}, building on earlier work by
Terauchi and Aiken~\cite{Terauchi2005}. The initial work discusses
verification but it does not offer an algorithm; numerous program
techniques have been subsequently proposed to verify hyperproperties.
Product programs are an alternative approach to relational
verification~\cite{Barthe+ProdVerification}. This approach can
leverage existing non-relational verification tools and techniques
when verifying the product program, but the large state space of
product programs can make verification difficult in practice. Product
programs have been used to verify $k$-safety properties and reason
about noninterference and secure information
flow~\cite{Barthe+SecSelfComp,Kovacs+CFGProd}. Barthe et
al.~\cite{barthe2013beyond} have developed a set of necessary
conditions for ``left-product programs''; these product programs can
be used to verify hyperproperties outside of $k$-safety, including our
$\forall\exists$ properties, although the work does not address how to
construct left-product programs.

Unno et al.~\cite{Unno2021} have developed a technique for verifying
\AEH{} properties including program refinement, generalized
noninterference, and cotermination by encoding a constraint
satisfaction problem expressed using a generalization of constrained
Horn clauses. The approach solves constraints using a stratified CEGIS
approach, and can synthesize non-trivial alignment predicates for
interleaving executions of loop bodies. This work is not based on a
Hoare-style program logic, but rather develops per-property embeddings
of \AEH{} verification problems in a novel adaptation of constrained
Horn clauses.

There are several modal logics which support a style of existential
reasoning similar to our existential logic. Temporal logics like
HyperLTL and HyperCTL~\cite{Clarkson2014} can be used to reason about
hyperproperties, although verification tooling~\cite{Clarke1994} is
focused on model checking state transition systems rather than program
logics. Coenen et al.~\cite{Coenen2019} examine verification and
synthesis of computational models using HyperLTL formulas with
alternating quantifiers. Cook et al.~\cite{Cook2013} examine
existential reasoning in branching-time temporal logics by way of
removing state space until universal reasoning methods can be used.
Lamport and Schneider~\cite{Lamport2021} examine using TLA to verify
\AEH{} properties including refinement and GNI. While the above
approaches are capable of reasoning about the kinds of liveness
properties we consider in this paper, they all focus on model checking
state transition systems rather than using a Hoare-style logic to
reason directly over programs as in our approach.

%% file: sections/08-conclusion.tex
\section{Conclusion}
\label{sec:Conclusion}
\vspace{-.7em} This paper presented \rhle{}, a novel relational
Hoare-style program logic for reasoning about \AEH{} properties. These
properties can capture a variety of interesting behaviors of multiple
program executions, including program refinement and information flow
properties. Key to our logic is a novel form of function
specifications which constrain the set of behaviors that a valid
implementation of a function \emph{must} exhibit. We have developed an
automated verification algorithm based on \rhle{}, and we demonstrated
that an implementation of this algorithm is able to check the validity
of a variety of \AEH{} properties over a benchmark suite of programs.

%% file: sections/99-appendix.tex
\newpage
\appendix
\section{Semantics of \FunImp{}}
\label{sec:FunImpSemantics}

The semantics of \FunImp{} is given as a big-step reduction relation
from initial to final states. This relation is parameterized over an
\emph{implementation} context $I \in \mathcal{N} \rightarrow
\mathit{FD}$, a partial mapping from function names to definitions.
\FunImp{} program states, $\sigma \in \mathcal{V} \rightarrow
\mathbb{N}$, are mappings from variables to their current value. The
reduction relation is also parameterized over an interpretation used
to determine the validity of assertions; we write $\sigma \models P$
to denote that the assertion $P$ holds in state $\sigma$. We condense
sequences of repeated expressions in a similar way to function
arguments, writing the sequence
$I \vdash \sigma, a_1 \Downarrow v_1 \; \cdots \;
    I \vdash \sigma, a_n \Downarrow v_n$ as
$I \vdash \sigma, \overline{a} \Downarrow \overline{v}$ and
$[ x_1 \mapsto v_1, \ldots, x_n \mapsto v_n]$ as
$[ \overline{x} \mapsto \overline{v}]$, for example.

The evaluation rules of \FunImp{} are presented in
\autoref{fig:FunSemantics}.

\begin{figure}
  \begin{mathpar}
    \footnotesize
    \inferrule*[right=ESkip]
    {
    }
    {
      I \vdash \sigma, \sskip \Downarrow \sigma
    }

    \inferrule*[right=EAssgn]
    {
      \sigma, a \Downarrow v
    }
    {
      I \vdash \sigma, \asgn{x}{a} \Downarrow [x \mapsto v]\sigma
    }

    \inferrule*[right=EHavoc]
    {
    }
    {
      I \vdash \sigma, x ::= {\textsf{havoc}} \Downarrow [x \mapsto v]\sigma
    }

    \inferrule*[right=ESeq]
    {
      I \vdash \sigma, c_1 \Downarrow \sigma' \\
      I \vdash \sigma', c_1 \Downarrow \sigma''
    }
    {
      I \vdash \sigma, \seq{c_1}{c_2} \Downarrow \sigma''
    }

    \inferrule*[right=ECondT]
    {
      \sigma, b \Downarrow \textsf{true} \\
      I \vdash \sigma, c_1 \Downarrow \sigma'
    }
    {
      I \vdash \sigma,\ite{b}{c_1}{c_2} \Downarrow \sigma'
    }

    \inferrule*[right=ECondF]
    {
      \sigma, b \Downarrow \mathtt{\bot} \\
      I \vdash \sigma, c_2 \Downarrow \sigma'
    }
    {
      I \vdash \sigma,\ite{b}{c_1}{c_2} \Downarrow \sigma'
    }

    \inferrule*[right=EWhileT]
    {
      \sigma, b \Downarrow \textsf{true} \\
      I \vdash \sigma, c \Downarrow \sigma' \\\\
      I \vdash \sigma', \while{b}{c} \Downarrow \sigma''
    }
    {
      I \vdash \sigma, \while{b}{c} \Downarrow \sigma''
    }

    \inferrule*[right=EWhileF]
    {
      \sigma, b \Downarrow \textsf{false} \\
    }
    {
      I \vdash \sigma, \while{b}{c} \Downarrow \sigma
    }

  \inferrule*[right=ECall] {
    I(f) = \textsf{\textbf{def} } f(\overline{x})\, \{s; \textsf{\textbf{return} } e \} \\
    I \vdash \sigma, \overline{a} \Downarrow \overline{v} \\
    I \vdash [ \overline{x} \mapsto \overline{v}], s \Downarrow \sigma' \\
    I \vdash \sigma', e \Downarrow r \\
  }
  {
    I \vdash \sigma, y := f(\overline{a}) \Downarrow
    [y \mapsto r]\sigma
  }
  \end{mathpar}
  \caption{Big-step evaluation relation of \FunImp{} with a concrete
  implementation context.}
  \label{fig:FunSemantics}
\end{figure}

\subsection{Overapproximate Executions Semantics}

The big-step operational semantics for overapproximate evaluation is
given in \autoref{fig:UniversalExecSemantics}. These semantics are
nearly identical to the evaluation semantics over concrete
implementation contexts given in \autoref{fig:FunSemantics}, but are
instead parameterized over a universal specification context
$S_\forall \in \mathcal{N} \rightarrow FA$ and replaces the
\textsc{ECall} rule with the \textsc{ECall}$_\forall$ rule. The latter
rule allows a call to a universally specified function to step to any
state with a return value consistent with the function's
specification.

\begin{figure}
  \begin{mathpar}
    \footnotesize
    \inferrule*[right=ESkip$_\forall$]
    {
    }
    {
      S_\forall \vdash \sigma, \sskip \Downarrow_\forall \sigma
    }

    \inferrule*[right=EAssgn$_\forall$]
    {
      \sigma, a \Downarrow_\forall v
    }
    {
      S_\forall \vdash \sigma, \asgn{x}{a} \Downarrow_\forall [x \mapsto v]\sigma
    }

    \inferrule*[right=EHavoc$_\forall$]
        {  }
        {
          S_\forall \vdash \sigma, x:=\textsf{havoc} \Downarrow_\forall [x \mapsto v]\sigma
        }

    \inferrule*[right=ESeq$_\forall$]
    {
      S_\forall \vdash \sigma, c_1 \Downarrow_\forall \sigma' \\
      S_\forall \vdash \sigma', c_1 \Downarrow_\forall \sigma''
    }
    {
      S_\forall \vdash \sigma, \seq{c_1}{c_2} \Downarrow_\forall \sigma''
    }

    \inferrule*[right=ECondT$_\forall$]
    {
      \sigma, b \Downarrow_\forall \textsf{true} \\
      S_\forall \vdash \sigma, c_1 \Downarrow_\forall \sigma'
    }
    {
      S_\forall \vdash \sigma,\ite{b}{c_1}{c_2} \Downarrow_\forall \sigma'
    }

    \inferrule*[right=ECondF$_\forall$]
    {
      \sigma, b \Downarrow_\forall \mathtt{\bot} \\
      S_\forall \vdash \sigma, c_2 \Downarrow_\forall \sigma'
    }
    {
      S_\forall \vdash \sigma,\ite{b}{c_1}{c_2} \Downarrow_\forall \sigma'
    }

    \inferrule*[right=EWhileT$_\forall$]
    {
      \sigma, b \Downarrow_\forall \textsf{true} \\
      S_\forall \vdash \sigma, c \Downarrow_\forall \sigma' \\\\
      S_\forall \vdash \sigma', \while{b}{c} \Downarrow_\forall \sigma''
    }
    {
      S_\forall \vdash \sigma, \while{b}{c} \Downarrow_\forall \sigma''
    }

    \inferrule*[right=EWhileF$_\forall$]
    {
      \sigma, b \Downarrow_\forall \textsf{false} \\
    }
    {
      S_\forall \vdash \sigma, \while{b}{c} \Downarrow_\forall \sigma
    }

  \inferrule*[right=ECall$_{\forall 1}$]
  {
    S_{\forall}(f) = \textsf{ax}_{\forall}\; f(\overline{x})\, \{P\}\, \{Q\} \\
    S_{\forall} \vdash \sigma, \overline{a} \Downarrow_\forall \overline{v} \\
    [ \overline{x} \mapsto \overline{v}] \models P  \\
    [ \rho \mapsto r, \overline{x} \mapsto \overline{v}] \models Q  \\
  }
  {
    S_{\forall} \vdash \sigma, y := f(\overline{a}) \Downarrow_\forall [y \mapsto r]\sigma
  }

  \inferrule*[right=ECall$_{\forall 2}$]
  {
    S_{\forall}(f) = \textsf{ax}_{\forall}\; f(\overline{x})\, \{P\}\, \{Q\} \\
    S_{\forall} \vdash \sigma, \overline{a} \Downarrow_\forall \overline{v} \\
    [ \overline{x} \mapsto \overline{v}] \not \models P  \\
  }
  {
    S_{\forall} \vdash \sigma, y := f(\overline{a}) \Downarrow_\forall [y \mapsto r]\sigma
  }
  \end{mathpar}
  \caption{Overapproximate execution semantics of \FunImp{} with a
  universal specification context.}
  \label{fig:UniversalExecSemantics}
\end{figure}

\section{Universal Hoare Logic}
\label{sec:HL}
\begin{figure}[H]
  \begin{mathpar}
    \small
    \inferrule*[right=$\forall$Conseq]
    {
      \models P \implies P' \\
      \models Q' \implies Q \\
      \hltrip{P'}{c}{Q'}
    }
    {
      \hltrip{P}{c}{Q}
    }

    \inferrule*[right=$\forall$Skip]
    {
    }
    {
      \hltrip{P}{\sskip}{P}
    }

    \inferrule*[right=$\forall$Assgn]
    {
    }
    {
      \hltrip{P[a/x]}{\asgn{x}{a}}{P}
    }

    \inferrule*[right=$\forall$Havoc]
    {
    }
    {
      \hltrip{\forall v. P[v/x]}{x ::= \textsf{havoc}}{P}
    }

    \inferrule*[right=$\forall$Seq]
    {
      \hltrip{P}{c_1}{P'} \\
      \hltrip{P'}{c_2}{Q}
    }
    {
      \hltrip{P}{\seq{c_1}{c_2}}{Q}
    }

    \inferrule*[right=$\forall$Cond]
    {
      \hltrip{P \land b}{c_1}{Q} \\
      \hltrip{P \land \neg b}{c_2}{Q}
    }
    {
      \hltrip{P}{\ite{b}{c_1}{c_2}}{Q}
    }

    \inferrule*[right=$\forall$While]
    {
      \hltrip{P\; \land\; b}{c}{P}
    }
    {
      \hltrip{P}{\while{b}{c}}{P\; \land\; \neg b}
    }

    \inferrule*[right=$\forall$Spec]
    {
      S(f)= \textsf{ax } f (\overline{x}) \{P\} \{Q\}
    }
    {
      \hltrip{
        \begin{minipage}{.34\columnwidth}
          \begin{align*}
            P [\overline{a}/\overline{x}]
            & \land~~ \forall v. Q [v/\rho; \overline{a}/\overline{x}]
            & \implies Q[v/y]
          \end{align*}
        \end{minipage}
      }
      {\asgn{y}{f(\overline{a})}}{Q}
    }
  \end{mathpar}
  \caption{Proof rules for a universal Hoare logic for \FunImp{}.}
  \label{fig:HLRules}
\end{figure}

\newpage

\section{Existential Hoare Logic}
\label{sec:FullExistentialRules}

\begin{figure}[H]
  \begin{mathpar}
    \footnotesize
    \inferrule*[right=$\exists$Conseq]
    {
      \models P \implies P' \\
      \models Q' \implies Q \\
      \hletrip{P'}{s}{Q'}
    }
    {
      \hletrip{P}{s}{Q}
    }

    \inferrule*[right=$\exists$Skip]
    {
    }
    {
      \hletrip{P}{\sskip}{P}
    }

    \inferrule*[right=$\exists$Assgn]
    {
    }
    {
      \hletrip{Q[a/x]}{\asgn{x}{a}}{Q}
    }

    \inferrule*[right=$\exists$Havoc]
    {
    }
    {
      \hletrip{\exists v.\; Q[v/x]}{\asgn{x}{\textsf{havoc}}}{Q}
    }

    \inferrule*[right=$\exists$Seq]
    {
      \hletrip{P}{s_1}{P'} \\
      \hletrip{P'}{s_2}{Q}
    }
    {
      \hletrip{P}{\seq{s_1}{s_2}}{Q}
    }

    \inferrule*[right=$\exists$Cond]
    {
      \hletrip{P \land b}{s_1}{Q} \\
      \hletrip{P \land \neg b}{s_2}{Q}
    }
    {
      \hletrip{P}{\ite{b}{s_1}{s_2}}{Q}
    }

    \inferrule*[right=$\exists$While]
    { R \textsf{ is well-founded} \\
      \hletrip{P\; \land\; b\; \land\; M\, a }
      {s}
      {P \;\land\; \exists a'.\: M\, a' \;\land\; a'\, R\, a}
    }
    {
      \hletrip{P \land \exists a.\: M\, a}{\while{b}{s}}{P\; \land\; \neg b}
    }

    \inferrule*[right=$\exists$Spec]
    {
      S_\exists(f) = \textsf{ax}_\exists f(\overline{x})\; [\overline{c}]\; \{P\}\; \{Q\}
    }
    {
      \hletrip{
        \begin{minipage}{.34\columnwidth}
          \begin{align*}
            \exists \overline{k}.
            \;(
            &[\overline{x} \mapsto \overline{a}] \models P[\overline{k}/\overline{c}] \\
            ~~ \land ~~ & \exists v.
                          [\rho \mapsto v, \overline{x} \mapsto
                          \overline{a}] \models Q[\overline{k}/\overline{c}]
            \\
            ~~\land~~ &  \forall v.
                        [\rho \mapsto v, \overline{x} \mapsto
                        \overline{a}] \models
                        Q[\overline{k}/\overline{c}] \\
            & \quad\quad \implies
                        R[v/y]
                        )
          \end{align*}
        \end{minipage}
      }
      {\asgn{y}{f(\overline{a})}}{R}
    }

  \end{mathpar}
  \caption{Existential Hoare logic rules.}
  \label{fig:EHLRules}
\end{figure}

\newpage

\section{Synchronous Rules}
\label{sec:SynchronousRules}

\begin{figure}[H]
  \begin{mathpar}
    \small


    \inferrule*[right=SkipIntroL]
    {
      \rhletrip{\Phi}
      {\overline{\seq{\sskip}{p_\forall}}}
      {\overline{\seq{\sskip}{p_\exists}}}
      {\Psi}
    }
    {
      \rhletrip{\Phi}
      {\overline{p_\forall}}
      {\overline{p_\exists}}
      {\Psi}
    }

    \inferrule*[right=SyncSeq]
    {
      S_\forall, S_\exists \vdash \epre{\Phi}
      \overline{s_\forall} \esim
      \overline{s_\exists}
      \epost{\chi} \\
      S_\forall, S_\exists \vdash \epre{\chi}
      \overline{s_\forall'} \esim
      \overline{s_\exists'}
      \epost{\Psi}
    }
    {
      S_\forall, S_\exists \vdash \epre{\Phi}
      \overline{\seq{s_\forall}{s_\forall'}} \esim
      \overline{\seq{s_\exists}{s_\exists'}}
      \epost{\Psi}
    }

    \inferrule*[right=SyncLoops]
    {
      \rhletrip{\mathbb{I} \land \bigwedge_{0 \le i \le n} b_i}
               {s_0 \circledast \cdots \circledast s_k}{s_{k+1}
                 \circledast \cdots \circledast s_n}
               {\mathbb{I}} \\
      \mathbb{I} \land \bigwedge_{0 \le i \le n} \neg b_i \implies \Psi \\
      \mathbb{I} \land \neg \bigwedge_{0 \le i \le n} b_i \implies \bigwedge_{0 \le i \le n} \neg b_i
    }
    {
      S_\forall, S_\exists \vdash \epre{\mathbb{I}}
      \while{b_0}{s_0} \circledast \cdots \circledast \while{b_k}{s_k} \\
      \phantom{S_\forall, S_\exists \vdash} \esim
      \while{b_{k+1}}{s_{k+1}} \circledast \cdots \circledast \while{b_n}{s_n}
      \epost{\Psi}
    }

    \inferrule*[right=SyncLoops$_\exists$]
    {
      R \textsf{ is well-founded} \\
      \rhletrip{\mathbb{I} \land \bigwedge_{0 \le i \le n} b_i \land M\; a}
               {\overline{\sskip}}{s_0 \circledast \cdots \circledast s_n}
               {\mathbb{I} \land \exists a'.\: M\, a' \;\land\; a'\, R\, a} \\
               \mathbb{I} \land \bigwedge_{0 \le i \le n} \neg b_i \implies \Psi \\
               \mathbb{I} \land \neg \bigwedge_{0 \le i \le n} b_i \implies \bigwedge_{0 \le i \le n} \neg b_i
    }
    {
      S_\forall, S_\exists \vdash \epre{\mathbb{I} \land \exists a.\:
        M\, a }
      \overline{\sskip} \esim
      \while{b_0}{s_0} \circledast \cdots \circledast \while{b_n}{s_n}
      \epost{\Psi}
    }

    \inferrule*[right=SyncLoops$_\forall$]
    {
      \rhletrip{\mathbb{I} \land \bigwedge_{0 \le i \le n} b_i}
      {s_0 \circledast \cdots \circledast s_n}{\overline{\sskip}}
      {\mathbb{I}} \\
      \mathbb{I} \land \bigwedge_{0 \le i \le n} \neg b_i \implies \Psi \\
      \mathbb{I} \land \neg \bigwedge_{0 \le i \le n} b_i \implies \bigwedge_{0 \le i \le n} \neg b_i
    }
    {
      S_\forall, S_\exists \vdash \epre{\mathbb{I} \land \exists a.\:
        M\, a }
      \while{b_0}{s_0} \circledast \cdots \circledast
      \while{b_n}{s_n} \esim \overline{\sskip}
      \epost{\Psi}
    }

  \end{mathpar}
  \caption{Synchronous \rhle{} proof rules.}
  \label{fig:SyncRHLERules}
\end{figure}

\section{Proofs}
\label{sec:Metatheory}
\begin{theorem}
  \label{thm:ACompatFull}
  When run under an implementation context $I$ that is
  $\forall$-compatible with specification context $S_{\forall}$ with
  an initial state $\sigma$, a program $p$ will either diverge or
  evaluate to a state $\sigma'$ which is also the result of one of its
  overapproximate executions under $S_{\forall}$:
  \begin{align*}
    I \models_\forall S_\forall
    ~~\land~~ & I \vdash \sigma, p \Downarrow \sigma' ~~\implies~~
                S_\forall \vdash \sigma, p \Downarrow_\forall \sigma'
  \end{align*}
\end{theorem}
\begin{proof}
  By induction over the derivation of
  $I \vdash \sigma, p \Downarrow \sigma'$. The only interesting case
  is \textsc{ECall}, where we must consider whether $\sigma$ meets
  the precondition of $f$ in $S_\forall$. If not, the proof is
  immediate from \textsc{ECall$_{\forall 2}$}. If so, the proof
  follows from the fact that the definition of $f$ in $I$ is
  compatible with its specification in $S_\forall$ and
  \textsc{ECall$_{\forall 1}$}.
\end{proof}

\begin{theorem}
  \label{thm:ECompatFull}
  If there is an underapproximate evaluation of program $p$ to a set
  of states $\Sigma$ from an initial state $\sigma$ under
  $S_{\exists}$, then $p$ must terminate in at least one final state
  $\sigma' \in \Sigma$ from $\sigma$ under every implementation
  context $I$ that is $\exists$-compatible with $S_{\exists}$:
  \begin{align*}
    S_\exists \vdash \sigma, p \productive \Sigma
    ~~\land~~ I \models_\exists S_\exists ~~\implies~~
    \exists \sigma'. I \vdash \sigma, p \Downarrow \sigma' \;\land\;
    \sigma' \in \Sigma
  \end{align*}
\end{theorem}
\begin{proof}
  By induction over the derivation of
  $S_\exists \vdash \sigma, p \productive \Sigma$. Once again, the
  interesting case is \textsc{ECall$_\exists$}, which follows
  immediately from the fact that $I$ is $\exists$-compatible with $S_\exists$.
\end{proof}

\begin{theorem}[RHLE is Sound]
  Suppose $\rhletrip{\Phi}{\overline{p_\forall}}{\overline{p_\exists}}{\Psi}$. Then,
  for any function context $I$ compatible with $S_\forall$ and
  $S_\exists$, any set of initial states $\overline{\sigma_\forall}$
  and $\overline{\sigma_\exists}$ satisfying $\Phi$, and every
  collection of final states $\overline{\sigma_\forall'}$ of
  $\overline{p_\forall}$, there must exist a collection of final
  states produced by $\overline{p_\exists}$ that, together with
  $\overline{\sigma_\forall'}$ satisfies the relational post-condition
  $\Psi$:
  \begin{align*}
  & \rhletrip{\Phi}{\overline{p_\forall}}{\overline{p_\exists}}{\Psi}
    \\
    \land & \forall I.\, I \models S_\forall \land I \models S_\exists \\
    \land & \forall \overline{\sigma_\forall}\, \overline{\sigma_\exists}.\:
            \overline{\sigma_\forall}, \overline{\sigma_\exists} \models \Phi
    \\
    \land & \forall \overline{\sigma_\forall'}.\:
            I \vdash \overline{\sigma_\forall}, \overline{p_\forall} ~\Downarrow~ \overline{\sigma_\forall'}
            \implies \\
    & \exists \overline{\sigma_\exists'}.\:
      I \vdash \sigma_\exists, \overline{p_\exists} ~\Downarrow~
      \overline{\sigma_\exists'}
      ~\land~ \overline{\sigma_\forall'},\; \overline{\sigma_\exists'}
      \models \Psi
\end{align*}
\end{theorem}
\begin{proof}
  We first prove a stronger property by induction on the triple
  $\rhletrip{\Phi}{\overline{p_\forall}}{\overline{p_\exists}}{\Psi}$,
  namely that there exist appropriate existential executions of $p_\exists$ for
  every collection of final states of $p_\forall$ produced by
  an overapproximate execution, for any set of
  initial states satisfying the precondition $\Phi$:
  \begin{align}
  \begin{split}
    & \forall \overline{\sigma_\forall}\, \overline{\sigma_\exists}.\:
      \overline{\sigma_\forall}\, \overline{\sigma_\exists} \models \Phi ~~\land~~
      \forall \overline{\sigma_\forall'}.\:
      S_\forall \vdash \overline{\sigma_\forall}, \overline{p_\forall}
      ~ \Downarrow_\forall~ \overline{\sigma_\forall'} \implies \\
    & S_\exists \vdash \sigma_\exists, \overline{p_\exists} \productive
      \{\overline{\sigma_\exists'} ~|~ \overline{\sigma_\forall'},
      \overline{\sigma_\exists'} \models \Psi \}
    \label{eqn:IH}
  \end{split}
  \end{align}
  \noindent By \autoref{thm:ACompat}, the fact that $I$ is
  $\forall$-compatible with $S_\forall$, and our assumption that
  $I \vdash \overline{\sigma_\forall}, \overline{p_\forall}
  \Downarrow \overline{\sigma_\forall'}$, it follows that:
  \begin{align}
    S_\forall \vdash \overline{\sigma_\forall}, \overline{p_\forall}
    \Downarrow_\forall \overline{\sigma_\forall'}
    \label{eqn:EmptyEval}
  \end{align}
  \noindent Armed with (\ref{eqn:IH}) and (\ref{eqn:EmptyEval}) and the
  assumption that $I$ is $\exists$-compatible with $S_\exists$, by
  \autoref{thm:ECompat} we can conclude the desired result, i.e.
  $\exists \overline{\sigma_\exists'}.\: I \vdash \sigma_\exists,
  \overline{p_\exists} \Downarrow \overline{\sigma_\exists'} ~\land~
  \overline{\sigma_\forall'},\; \overline{\sigma_\exists'} \models
  \Psi$.
\end{proof}

\section{Verification Example}
\label{sec:VerificationExample}
To illustrate the operation of \texttt{VCGen}, consider proving the
following simple refinement assertion, where the contexts $S_\forall$
and $S_\exists$ contain the specifications for \textsf{randB} from
\autoref{ex:RandBSpecs}:
\[
\rhletrip{\top}{\asgn{y_1}{\textsf{randB}(4)}}{\asgn{y_2}{\textsf{randB}(10)}}{y_1 = y_2}
\]
\texttt{RHLEVerify} begins by calling \texttt{VCGen} with:
\begin{align*}
p_{\forall} &\equiv \{\sskip{}; \asgn{y_1}{\textsf{randB}(4)} \} \\
p_{\exists} &\equiv \{\sskip{}; \asgn{y_2}{\textsf{randB}(10)} \} \\
\overline{\Psi} &\equiv (\varnothing, \varnothing, y_1 = y_2)
\end{align*}
\texttt{VCGen} matches the \textsf{randB} call in $p_\forall$ (line
5) and recurses with the new postcondition built by $wp_\forall$
(omitting the trivial precondition for brevity):
\begin{align*}
p_{\forall} &\equiv \{ \sskip{} \} \\
p_{\exists} &\equiv \{ \asgn{y_2}{\textsf{randB}(10)} \} \\
\overline{\Psi} &\equiv (\{ v_1 \}, \varnothing, 0 \le v_1 < 4 \implies v_1 = y_2)
\end{align*}
\texttt{VCGen} now chooses the existential call to \textsf{randB} (line 7), and
recurses again with a postcondition built by $wp_\exists$:
\begin{align*}
p_{\forall} &\equiv \{ \sskip{} \} \\
p_{\exists} &\equiv\{ \sskip{} \} \\
\overline{\Psi} &\equiv (\{ v_1 \}, \{ v_2 \}, 0 \le v_2 < 10 \wedge (0 \le v_1 < 4 \implies v_1 = v_2))
\end{align*}
Since both programs are now \sskip{}, \texttt{VCGen} terminates,
returning
$(\{ v_1 \}, \{ v_2 \}, 0 \le v_2 < 10 \wedge (0 \le v_1 < 4 \implies
v_1 = v_2))$.  \texttt{RHLEVerify} uses this formula to construct the
following query:
\[
\forall v_1\; \exists v_2.\; 0 \le v_2 < 10 \wedge (0 \le v_1 < 4 \implies v_1 = v_2)
\]
which it hands off to \texttt{Verify}. Note how this formula
encodes the essence of the \AEH{} question posed by the original
triple: for all allowed \textsf{randB} return values $v_1$ in the
universal execution, we want to know if there exists an allowed
instantiation of the choice variable $v_2$ in the existential
execution that brings both programs to the same final state. In this
case, the $\forall \exists$ formula is valid, and verification
succeeds.

\section{Example ORHLE Input}
\label{sec:orhleInput}
The following listing verifies a noninterference property, namely that
the program never leaks any information about the variable
\lstinline[style=funimp]|high|. Note that the underapproximation of
\lstinline|flipCoin| is required. If linked to a \lstinline|flipCoin|
implementation that always returns 0, for example, attackers could
always know whether or not the initial value of \lstinline|low| was
less than \lstinline|high| by observing \lstinline[style=funimp]|low|.

\begin{lstlisting}[style=orhle]
forall: run[1];
exists: run[2];

pre:  (= run!1!low run!2!low);
post: (= run!1!low run!2!low);

aspecs:
  flipCoin() {
    pre:  true;
    post: (or (= ret! 0) (= ret! 1));
  }

especs:
  flipCoin() {
    choiceVars: n;
    pre:  (or (= n 0) (= n 1));
    post: (= ret! n);
  }

fun run(high, low) {
  if (low < high) then
    low := 0;
  else
    low := 1;
  end
  flip := call flipCoin();
  if (flip == 0) then
    low := 1 - low;
  endif
}
\end{lstlisting}

Conversely, \orhle{} identifies a violation of noninterference in the
listing below. The program might leak the value of
\lstinline[style=funimp]|high|, depending on the outcome of
\lstinline|flipCoin|.

\begin{lstlisting}[style=orhle]
forall: run[1];
exists: run[2];

pre:  (= run!1!low run!2!low);
post: (= run!1!low run!2!low);

aspecs:
  flipCoin() {
    pre:  true;
    post: (or (= ret! 0) (= ret! 1));
  }

especs:
  flipCoin() {
    choiceVars: n;
    pre:  (or (= n 0) (= n 1));
    post: (= ret! n);
  }

fun run(high, low) {
  flip := call flipCoin();
  if (flip == 0) then
    low := high + low;
  endif
}
\end{lstlisting}

%% file: rhle.bbl
\begin{thebibliography}{10}
\providecommand{\url}[1]{\texttt{#1}}
\providecommand{\urlprefix}{URL }
\providecommand{\doi}[1]{https://doi.org/#1}

\bibitem{Abadi+Refinement}
{Abadi}, M., {Lamport}, L.: The existence of refinement mappings. In: [1988]
  Proceedings. Third Annual Symposium on Logic in Computer Science. pp.
  165--175 (1988)

\bibitem{Aguirre+RHOL}
Aguirre, A., Barthe, G., Gaboardi, M., Garg, D., Strub, P.Y.: A relational
  logic for higher-order programs. Proc. ACM Program. Lang.  \textbf{1}(ICFP),
  21:1--21:29 (Aug 2017)

\bibitem{VST}
Appel, A.W.: Verified software toolchain. In: Barthe, G. (ed.) Programming
  Languages and Systems. pp. 1--17. Springer Berlin Heidelberg, Berlin,
  Heidelberg (2011)

\bibitem{Banerjee2019}
Banerjee, A., Nagasamudram, R., Naumann, D.A., Nikouei, M.: A relational
  program logic with data abstraction and dynamic framing. arXiv preprint
  arXiv:1910.14560  (2019)

\bibitem{Barthe+ProdVerification}
Barthe, G., Crespo, J.M., Kunz, C.: Relational verification using product
  programs. In: Butler, M., Schulte, W. (eds.) FM 2011: Formal Methods. pp.
  200--214. Springer Berlin Heidelberg, Berlin, Heidelberg (2011)

\bibitem{barthe2013beyond}
Barthe, G., Crespo, J.M., Kunz, C.: Beyond 2-safety: Asymmetric product
  programs for relational program verification. In: International Symposium on
  Logical Foundations of Computer Science. pp. 29--43. Springer (2013)

\bibitem{Barthe+SecSelfComp}
Barthe, G., D'Argenio, P.R., Rezk, T.: Secure information flow by
  self-composition. Mathematical Structures in Computer Science
  \textbf{21}(6),  1207–1252 (2011)

\bibitem{Barthe+pRHL}
Barthe, G., Gr{\'e}goire, B., Zanella~B{\'e}guelin, S.: Formal certification of
  code-based cryptographic proofs. SIGPLAN Not.  \textbf{44}(1),  90--101 (Jan
  2009)

\bibitem{Benton+RHL}
Benton, N.: Simple relational correctness proofs for static analyses and
  program transformations. In: Proceedings of the 31st ACM SIGPLAN-SIGACT
  Symposium on Principles of Programming Languages. pp. 14--25. POPL '04, ACM,
  New York, NY, USA (2004)

\bibitem{Clarke1994}
Clarke, E., Grumberg, O., Long, D.: Verification tools for finite-state
  concurrent systems. In: de~Bakker, J.W., de~Roever, W.P., Rozenberg, G.
  (eds.) A Decade of Concurrency Reflections and Perspectives. pp. 124--175.
  Springer Berlin Heidelberg, Berlin, Heidelberg (1994)

\bibitem{Clarkson2014}
Clarkson, M.R., Finkbeiner, B., Koleini, M., Micinski, K.K., Rabe, M.N.,
  S{\'a}nchez, C.: Temporal logics for hyperproperties. In: International
  Conference on Principles of Security and Trust. pp. 265--284. Springer (2014)

\bibitem{Clarkson+HYP}
Clarkson, M.R., Schneider, F.B.: Hyperproperties. J. Comput. Secur.
  \textbf{18}(6),  1157--1210 (Sep 2010)

\bibitem{Coenen2019}
Coenen, N., Finkbeiner, B., Sánchez, C., Tentrup, L.: Verifying hyperliveness
  pp. 121--139 (07 2019)

\bibitem{Cook2013}
Cook, B., Koskinen, E.: Reasoning about nondeterminism in programs. In:
  Proceedings of the 34th ACM SIGPLAN conference on Programming language design
  and implementation. pp. 219--230 (2013)

\bibitem{OrhleZenodo}
Dickerson, R., Ye, Q., Zhang, M.K., Delaware, B.: {ORHLE} (2022).
  \doi{10.5281/zenodo.7058107}

\bibitem{OrhleBenchmarks}
Dickerson, R., Ye, Q., Zhang, M.K., Delaware, B.: {RHLE Benchmarks} (2022),
  \url{https://github.com/rcdickerson/rhle-benchmarks}

\bibitem{DilligAbductiveInv}
Dillig, I., Dillig, T., Li, B., McMillan, K.: Inductive invariant generation
  via abductive inference. In: Proceedings of the 2013 ACM SIGPLAN
  International Conference on Object Oriented Programming Systems Languages and
  Applications. p. 443–456. OOPSLA '13, Association for Computing Machinery,
  New York, NY, USA (2013)

\bibitem{Houdini}
Flanagan, C., Leino, K.R.M.: {Houdini, an Annotation Assistant for ESC/Java}.
  In: Proceedings of the International Symposium of Formal Methods Europe on
  Formal Methods for Increasing Software Productivity. p. 500–517. FME '01,
  Springer-Verlag, Berlin, Heidelberg (2001)

\bibitem{Hoare+69}
Hoare, C.A.R.: An axiomatic basis for computer programming. Commun. ACM
  \textbf{12}(10),  576–580 (Oct 1969)

\bibitem{Iris}
Jung, R., Jourdan, J.H., Krebbers, R., Dreyer, D.: Rustbelt: Securing the
  foundations of the rust programming language. Proc. ACM Program. Lang.
  \textbf{2}(POPL) (dec 2017)

\bibitem{Jung+Future}
Jung, R., Lepigre, R., Parthasarathy, G., Rapoport, M., Timany, A., Dreyer, D.,
  Jacobs, B.: The future is ours: Prophecy variables in separation logic. Proc.
  ACM Program. Lang.  \textbf{4}(POPL) (Dec 2019)

\bibitem{Kovacs+CFGProd}
Kovács, M., Seidl, H., Finkbeiner, B.: Relational abstract interpretation for
  the verification of 2-hypersafety properties. pp. 211--222 (11 2013)

\bibitem{Lam2019}
{Lam}, W., {Oei}, R., {Shi}, A., {Marinov}, D., {Xie}, T.: idflakies: A
  framework for detecting and partially classifying flaky tests. In: 2019 12th
  IEEE Conference on Software Testing, Validation and Verification (ICST). pp.
  312--322 (2019)

\bibitem{Lamport2021}
Lamport, L., Schneider, F.B.: Verifying {Hyperproperties} {With} {TLA}. In:
  2021 {IEEE} 34th {Computer} {Security} {Foundations} {Symposium} ({CSF}). pp.
  1--16 (Jun 2021), iSSN: 2374-8303

\bibitem{Mclean1996}
McLean, J.: A general theory of composition for a class of "possibilistic"
  properties. IEEE Trans. Softw. Eng.  \textbf{22}(1),  53–67 (Jan 1996)

\bibitem{Nagasamudram2021}
Nagasamudram, R., Naumann, D.A.: Alignment completeness for relational hoare
  logics. In: 2021 36th Annual ACM/IEEE Symposium on Logic in Computer Science
  (LICS). pp. 1--13 (2021)

\bibitem{Concur+Sep+Logic}
O’Hearn, P.W.: Resources, concurrency, and local reasoning. Theoretical
  Computer Science  \textbf{375}(1),  271--307 (2007), festschrift for John C.
  Reynolds’s 70th birthday

\bibitem{OHearn+Incorrectness}
O’Hearn, P.W.: Incorrectness logic. Proc. ACM Program. Lang.
  \textbf{4}(POPL) (Dec 2019)

\bibitem{padhi2016data}
Padhi, S., Sharma, R., Millstein, T.: Data-driven precondition inference with
  learned features. ACM SIGPLAN Notices  \textbf{51}(6),  42--56 (2016)

\bibitem{LoopInvGen}
Padhi, S., Sharma, R., Millstein, T.: {LoopInvGen: A Loop Invariant Generator
  based on Precondition Inference} (2017)

\bibitem{Hoare-Java}
Poetzsch-Heffter, Arndand~M{\"u}ller, P.: {A Programming Logic for Sequential
  Java}. In: Swierstra, S.D. (ed.) Programming Languages and Systems. pp.
  162--176. Springer Berlin Heidelberg, Berlin, Heidelberg (1999)

\bibitem{Pratt1976}
Pratt, V.R.: {Semantical consideration on Floyd-Hoare logic}. In: 17th Annual
  Symposium on Foundations of Computer Science (sfcs 1976). pp. 109--121. IEEE
  (1976)

\bibitem{Sep+Logic}
Reynolds, J.: Separation logic: a logic for shared mutable data structures. In:
  Proceedings 17th Annual IEEE Symposium on Logic in Computer Science. pp.
  55--74 (2002)

\bibitem{Shi2016}
{Shi}, A., {Gyori}, A., {Legunsen}, O., {Marinov}, D.: Detecting assumptions on
  deterministic implementations of non-deterministic specifications. In: 2016
  IEEE International Conference on Software Testing, Verification and
  Validation (ICST). pp. 80--90 (2016)

\bibitem{Sousa+CHL}
Sousa, M., Dillig, I.: Cartesian hoare logic for verifying k-safety properties.
  In: Proceedings of the 37th ACM SIGPLAN Conference on Programming Language
  Design and Implementation. pp. 57--69. PLDI '16, ACM, New York, NY, USA
  (2016)

\bibitem{Terauchi2005}
Terauchi, T., Aiken, A.: Secure information flow as a safety problem. In:
  Hankin, C., Siveroni, I. (eds.) Static Analysis. pp. 352--367. Springer
  Berlin Heidelberg, Berlin, Heidelberg (2005)

\bibitem{Unno2021}
Unno, H., Terauchi, T., Koskinen, E.: Constraint-{Based} {Relational}
  {Verification}. In: Silva, A., Leino, K.R.M. (eds.) Computer {Aided}
  {Verification}. pp. 742--766. Lecture {Notes} in {Computer} {Science},
  Springer International Publishing, Cham (2021)

\bibitem{Vries+ReverseHL}
de~Vries, E., Koutavas, V.: Reverse hoare logic. In: Proceedings of the 9th
  International Conference on Software Engineering and Formal Methods. p.
  155–171. SEFM’11, Springer-Verlag, Berlin, Heidelberg (2011)

\end{thebibliography}
